\documentclass[runningheads]{llncs}
\usepackage{latex/packages}
\usepackage{latex/macros_arrows}
\newcommand{\efs}[1][]{\linearrow[#1]}
\newcommand{\ebs}[1][]{\zigzagarrow[#1]}
\newcommand{\efbs}[1][]{\llinearrow[#1]}

\newcommand{\redl}[1]{\efs[#1]}

\newcommand{\revredl}[1]{\ebs[#1]}
\newcommand{\frevredl}[1]{\efbs[#1]}
\newcommand{\fwlts}[2]{\efs[\protect{#2[#1]}]}
\newcommand{\bwlts}[2]{\ebs[\protect{#2[#1]}]}
\newcommand{\fbwlts}[2]{\efbs[\protect{#2[#1]}]}

\newcommand{\congru}{\equiv}

\DeclareMathOperator{\rem}{rm}

\newcommand{\orig}[1]{O_{#1}}
\newcommand{\bs}{\backslash}

\newcommand{\dep}{\lessdot}

\newcommand{\Left}{\mathrm{L}}
\newcommand{\Right}{\mathrm{R}}
\renewcommand{\L}{\Left}
\newcommand{\R}{\Right}

\newcommand{\rev}{^{\bullet}}
\newcommand{\eq}{\sim}

\newcommand{\conc}{\smile}

\newcommand{\names}{\ensuremath{\mathsf{N}}}
\newcommand{\labels}{\ensuremath{\mathsf{L}}}
\newcommand{\enhanced}{\ensuremath{\mathsf{E}}}
\newcommand{\keys}{\ensuremath{\mathsf{K}}}
\newcommand{\out}[1]{\overline{#1}}

\newcommand{\st}{s.t.\ }

\newcommand*{\resp}{resp.\@\xspace}

\newcommand{\defeqc}{\stackrel{\text{def}}{=}}

\AtBeginDocument{

	\def\equationautorefname~#1\null{(#1)\null}
}

\newcommand*{\eg}{e.g.\@\xspace}

\newcommand*{\ie}{i.e.\@\xspace}

\newcommand*{\wrt}{w.r.t.\@\xspace}
\newcommand*{\wlg}{w.l.o.g.\@\xspace}

\makeatletter
\newcommand*{\etc}{%
    \@ifnextchar{.}%
        {etc}%
        {etc.\@\xspace}%
}
\makeatother

\newcommand{\ccs}{\textsc{CCS}\xspace}
\newcommand{\ccsd}{\textsc{CCS}\(_{\text{D}}\)\xspace}
\newcommand{\ccsr}{\textsc{CCS}\(_{\text{!}}\)\xspace}
\newcommand{\ccsi}{\textsc{CCS}\(_{\text{*}}\)\xspace}
\newcommand{\ccsk}{\textsc{CCSK}\xspace}
\newcommand{\rccs}{\textsc{RCCS}\xspace}

\DeclareMathOperator{\std}{std}
\DeclareMathOperator{\key}{key}
\DeclareMathOperator{\markr}{m}

\makeatletter
\DeclareFontEncoding{LS2}{}{\@noaccents}
\makeatother
\DeclareFontSubstitution{LS2}{stix}{m}{n}
\DeclareSymbolFont{largesymbolsstix}{LS2}{stixex}{m}{n}
\DeclareMathDelimiter{\lBrace}{\mathopen} {largesymbolsstix}{"E8}{largesymbolsstix}{"0E}
\DeclareMathDelimiter{\rBrace}{\mathclose}{largesymbolsstix}{"E9}{largesymbolsstix}{"0F}

\DeclareFontEncoding{LS1}{}{}
\DeclareFontSubstitution{LS1}{stix}{m}{n}
\DeclareSymbolFont{stixsymbols}{LS1}{stixscr}{m}{n}
\SetSymbolFont{stixsymbols}{bold}{LS1}{stixscr}{b}{n}
\DeclareMathSymbol{\kay}{\mathalpha}{stixsymbols}{"6B}

\newcommand{\Rev}[1]{\ensuremath{#1\rev}} %

\definecolor{Gray}{gray}{0.85}

\SetExtraKerning[unit=character]{encoding=*}{\textemdash={100,100}}

\makeatletter
\renewenvironment{description}%
{\list{}{\leftmargin=10pt %
		\labelwidth\z@ \itemindent-\leftmargin
		}}%
{\endlist}
\makeatother

\makeatletter
\newcommand{\citecomment}[2][]{\citen{#2}#1\citevar}
\newcommand{\citeone}[1]{\citecomment{#1}}
\newcommand{\citetwo}[2][]{\citecomment[,~#1]{#2}}
\newcommand{\citevar}{\@ifnextchar\bgroup{;~\citeone}{\@ifnextchar[{;~\citetwo}{]}}}
\newcommand{\citefirst}{\@ifnextchar\bgroup{\citeone}{\@ifnextchar[{\citetwo}{]}}}
\newcommand{\cites}{[\citefirst}
\makeatother

\title{Causal Consistent Replication in Reversible Concurrent Calculi}
\titlerunning{Causal Consistent Replication in Reversible Conc.\ Calculi}

\author{%
	Clément Aubert%
	\orcidID{0000-0001-6346-3043}
}
\institute{%
	School of Computer \& Cyber Sciences, Augusta University, Augusta, USA %
	\href{mailto:caubert@augusta.edu}{\email{caubert@augusta.edu}}, %
	\url{https://spots.augusta.edu/caubert/}
}

\begin{document}
	\maketitle
	\begin{abstract}
		Reversible computation is key in developing new, energy-efficient paradigms, but also in providing forward-only concepts with broader definitions and finer frames of study.
		Among other fields, the algebraic specification and representation of networks of agents have been greatly impacted by the study of reversible phenomena: reversible declensions of the calculus of communicating systems (\ccsk and \rccs) offer new semantic models, finer congruence relations, original properties, and revisits existing theories and results in a finer light.
		However, much remains to be done: concurrency, a central notion in establishing \emph{causal consistency}--a crucial property for reversibly systems--, was never given a clear and syntactical definition in \ccsk.
		While recursion was mentioned as a possible mechanism to inject infinite behaviors into the systems, replication was never studied.
		This work offers a solution to both problems, by leveraging a definition of concurrency developed for forward-only calculi using proved transition systems, by endowing \ccsk with a replication operator, and by studying the interplay of both notions.
		The system we obtain is the first reversible system capable of representing infinite behaviors that enjoys causal consistency, for our simple and purely syntactical notion of reversible concurrency.
		
	\end{abstract}

	\section{Introduction: Reversibility, Concurrency--Interplays}
	
	\textbf{Reversibility} has recently being dubbed \enquote{the future of life on earth}~\cite{nsf_reversibility} in one of the \href{https://www.nsf.gov/}{National Science Foundation}'s \emph{Idea Machine}, as it \textquote{promote\textins{s} a paradigm shift in the framing of research questions}, trying to locate \textquote{tipping points for irreversibility}.
	The study of bi-directional computation led to breakthroughs such as the design of CMOS adiabatic circuits~\cite{Frank2020} that will leverage Landauer's principle~\cite{doi:10.1126/sciadv.1501492,PhysRevLett.113.190601} to obtain \textquote{a greater level of energy efficiency than any other semiconductor\textins{\ldots} known today}.
	It also led to revisit established fields of computer science to obtain more general definitions, models, or tools, \eg regarding complexity~\cite{Lange2000,Paolini2018}, category~\cite{Kaarsgaard2017}, and automata~\cite{Axelsen2012,Holzer2017} theories, or \enquote{time-travel} debugging~\cite{Visan2011}, to name a few.
	
	\textbf{Concurrency Theory} has been re-shaped by reversibility as well: fine distinctions between causality and causation~\cite{Phillips2006} contradicted Milner's expansion laws~\cite[Example 4.11]{Lanese2021}, and the study of causal models for reversible computation led to novel correction criteria for causal semantics--both reversible and irreversible~\cite{Cristescu2015b}.
	\enquote{Traditional} equivalence relations have been captured syntactically~\cite{Aubert2020b}, while original observational equivalences were developed~\cite{Lanese2021}: reversibility triggered a global reconsideration of established theories and tools, with the clear intent of providing actionable methods for reversible systems~\cite{Lanese2018}, novel axiomatic foundations~\cite{Lanese2020} and original non-interleaving models~\cite{Aubert2015d,Graversen2021,Cristescu2015b}.
	
	\textbf{Two Systems} extend the Calculus of Communicating Systems (CCS)~\cite{Milner1980}--the godfather of \(\pi\)-~\cite{Sangiorgi2001}, Ambient~\cite{Zappa2005}, applied~\cite{Abadi2018} and distributed~\cite{Hennessy2007} calculi, among others--with reversible features.
	Reversible CCS (RCCS)~\cite{Danos2004} and CCS with keys (CCSK)~\cite{Phillips2006} are similarly the source of most~\cite{Arpit2017,Cristescu2015b,Medic2020,Mezzina2017}--if not all--of later systems developed to enhance reversible systems with some respect (rollback operator, name-passing abilities, probabilistic features).
	Even if those two systems share a lot of similarities~\cite{Lanese2019}, they diverge in some respects that are not fully understood--typically, it seems that different notions of \enquote{contexts with history} led to establish the existence of congruences for \ccsk~\cite[Proposition 4.9]{Lanese2021} or the impossibility thereof for \rccs~\cite[Theorem 2]{Aubert2021d}.
	However, they also share some shortcomings, and we offer to tackle two of them for \ccsk: a clear, syntactical, definition of concurrency that coincides with its forward-only counterpart, and the capacity of representing infinite behaviors.
	
	\textbf{Reversible Concurrency} is of course a central notion in the study of \rccs and \ccsk, as it enables the definition of \enquote{causal consistency}--a principle that, intuitively, states that backward reductions can undo an action only if its consequences have already been undone--and to obtain models where concurrency and causation are decorrelated~\cite{Phillips2007b}.
	As such, it have been studied from multiple angles, but, in our opinion, never in a fully satisfactory manner.
	Ad-hoc definitions relying on memory inclusion~\cite[Definition 3.11]{Krivine2006} or disjointness~\cite[Definition 7]{Danos2004} for \rccs, and semantical notions for both \rccs~\cite{Aubert2015d,Aubert2016jlamp,Aubert2020b} and \ccsk~\cite{Graversen2021,Phillips2007,Ulidowski2014} have been proposed, but, to our knowledge, have never been \begin{enumerate*}
		\item compared to each other,
		\item proven to be extensions of pre-existing forward-only definitions of concurrency.
	\end{enumerate*}
	Our contribution starts by introducing the first purely syntactical definition of concurrency for \ccsk, by extending the \enquote{universal} concurrency developed for forward-only \ccs~\cite{Degano2003}, that leveraged \emph{proved} transition systems~\cite{Degano2001}.
	
	\textbf{Infinite Behaviors} can be captured in \ccs using recursion (\ccsd), replication (\ccsr) or iteration, three strategies that are not equivalent~\cite{Palamidessi2005} \wrt decidability of divergence~\cite{Giambiagi2004} or termination~\cite{Busi2009}.
	Causal semantics for \ccsr are scarce: aside from a multiset semantics for the \(\pi\)-calculus with replication~\cite{Engelfriet1999}, the only work~\cite{Degano2003} we are aware of that addresses \ccsr's causal semantics remained unnoticed, despite its interest: it leverages an original notion of concurrency that was proven to be \enquote{universal} as it conservatively extends or agrees with pre-existing causal semantics of \ccs and \ccsd~\cite[Theorem 1]{Degano2003}.
	In reversible concurrent systems, only recursion have been investigated~\cite{Graversen2021,Krivine2006} or mentioned~\cite{Danos2004,Danos2005}, but without discussing causal consistency (with one exception~\cite[Lemma F.3]{Graversen2021}).
	Our contribution endows \ccsk with the first reversible notion of replication, and studies expected properties--such as the square diamond--as well as \enquote{reversible properties} for this infinite, reversible and concurrent system.
	
	\textbf{Our Contribution} is in two parts: after briefly recalling \ccs, its extensions to infinite behaviors and \ccsk (\autoref{sec:ccs+ccsk}), we define and study in detail our new causal semantics (\autoref{ssec:proved-ccsk}).
	We make crucial use of the loop lemma (\autoref{lemma:loop}) to define concurrency between coinitial traces in terms of concurrency between composable traces--a simple but ingenious mechanism that considerably reduces the definition and proof burdens.
	We then prove the \enquote{usual} sanity checks--the square and diamond properties (\autoref{sec:diamond})--and discuss other important properties--among which causal consistency (\autoref{sec:sanity}).
	
	We then discuss the interactions between extensions to infinite behavior and reversibility (\autoref{ssec:discussion}), and justify our choice of endowing \ccsk with replication (\autoref{ssec:repl-ccsk}).
	Our definition of concurrency is then extended to accommodate infinite behaviors, and most of our sanity checks are still valid (\autoref{ssec:square-inf}), even if our processes can now sometimes have two conflicting backward reductions--a phenomenon not observed in any other reversible system, to our knowledge.
	This oddity is then taken care of thanks to a simple restriction on identifiers, and causal consistency is restored (\autoref{ssec:causal-infinite}).%
	
	Our construction is built from simple definitions that are originally combined to obtain elegantly results--typically, the square property is derived from what we called the forward and sideways diamonds by leveraging the loop lemma--in the vein of the \enquote{axiomatic approach} to reversibility~\cite{Lanese2020}.
	We also briefly discuss structural congruence (\autoref{rem:congruence}), and sketch possible strategies to study infinite behaviors that duplicate their past.
	All proofs are included in \autoref{sec:proofs} for completeness, with their main arguments given in the body of the paper.
	A small technical lemma, intuitively explained in the body of the paper, is added in \autoref{app:sub-term} for completeness.

	\section{Infinite and Reversible Process Calculi}
	\label{sec:ccs+ccsk}
	\subsection{CCS with Infinite Behaviors}
	\subsubsection{Finite Core \ccs}
	We briefly recall the (forward-only) \enquote{finite fragment} of the core of \ccs (simply called \ccs) following a standard presentation~\cite{Busi2009}.
	The three different strategies to extend \ccs with infinite behavior--excluding recursion expressions and parametric definitions~\cite{Giambiagi2004}--will come in a second moment.
	
	\begin{definition}[(Co-)names and labels]
		Let \(\names=\{a,b,c,\dots\}\) be a set of \emph{names} and \(\out{\names}=\{\out{a},\out{b},\out{c},\dots\}\) its set of \emph{co-names}.
		The set of \emph{labels} \(\labels\) is \(\names \cup \out{\names} \cup\{\tau\}\), and we use \(\alpha\), \(\beta\) (\resp \(\lambda\)) to range over \(\labels\) (\resp \(\labels \bs \{\tau\}\)).
		A bijection \(\out{\cdot}:\names \to \out{\names}\), whose inverse is also written \(\out{\cdot}\), 
		gives the \emph{complement} of a name.%
	\end{definition}
	
	\begin{definition}[Operators]
		\label{def:operators}
		\ccs processes are defined as usual:
		\begin{multicols}{2}
			\noindent
			\begin{align}
				P,Q  \coloneqq  
				& P \bs \alpha \tag{Restriction} \\
				& P + Q  \tag{Sum}               
			\end{align}
			\begin{align}
				& \alpha. P \tag{Prefix}              \\
				& P \mid Q \tag{Parallel composition} 
			\end{align}
		\end{multicols}
		The inactive process \(0\) is omitted when preceded by a prefix, and the biding power of the operators, from highest to lowest, is 
		\(\bs \alpha\), \(\alpha.\), \(!\) (replication, introduced next), \(+\) and \(\mid\), 
		so that \eg \(\alpha . P + Q \bs \alpha \mid ! P + a\) is to be read as \(
		(
		(\alpha . P) + (Q\bs \alpha)
		)
		\mid ((!P) + (a.0))\).
		In a process \(P \mid Q\), we call \(P\) and \(Q\) \emph{threads}.
	\end{definition}
	
	The labeled transition system (LTS) for \ccs, denoted \(\redl{\alpha}\), is reminded in \autoref{fig:ltsrules}, where the \enquote{Infinite Group} should be ignored for now.
	\begin{figure}[h]
		\begin{tcolorbox}[title = Action and Restriction]
			\hfill
			\begin{prooftree}
				\hypo{}
				\infer[]1[act.]{\alpha. P \redl{\alpha}  P}
			\end{prooftree}
			\hfill
			\begin{prooftree}
				\hypo{ P   \redl{\alpha}P '}
				\infer[left label={\(a \notin \{\alpha, \out{\alpha}\}\)}]1[res.]{P  \bs a  \redl{\alpha}P ' \bs a}
			\end{prooftree}
			\hfill~
		\end{tcolorbox}
		
		\begin{tcolorbox}[adjusted title=Parallel Group]
			\hfill
			\begin{prooftree}
				\hypo{P \redl{\alpha}P'}
				\infer%
				1[\(\mid_{\L}\)]{P \mid  Q \redl{\alpha} P' \mid Q}
			\end{prooftree}
			\hfill
			\begin{prooftree}
				\hypo{P \redl{\lambda}  P'}
				\hypo{Q \redl{\out{\lambda}}  Q'}
				\infer%
				2[syn.]{P \mid Q \redl{\tau}  P' \mid Q'}
			\end{prooftree}
			\hfill
			\begin{prooftree}
				\hypo{Q \redl{\alpha}Q'}
				\infer%
				1[\(\mid_{\R}\)]{P \mid  Q  \redl{\alpha}P \mid Q'}
			\end{prooftree}
			\hfill
		\end{tcolorbox}
		\begin{tcolorbox}[adjusted title=Sum Group]
			\hfill
			\begin{prooftree}
				\hypo{P \redl{\alpha} P'}
				\infer1[\(+_{\L}\)]{Q + P \redl{\alpha} P'}
			\end{prooftree}
			\hfill
			\begin{prooftree}
				\hypo{P \redl{\alpha} P'}
				\infer1[\(+_{\R}\)]{Q + P  \redl{\alpha} P'}
			\end{prooftree}
			\hfill~
		\end{tcolorbox}
		\begin{tcolorbox}[adjusted title=Infinite Group]
			\begin{prooftree}
				\hypo{P \redl{\alpha}P'}
				\hypo{D \defeqc P}
				\infer[]2[const.]{D \redl{\alpha} P'}
			\end{prooftree}
			\hfill
			\begin{prooftree}
				\hypo{P \mid !P \redl{\alpha}  P'}
				\infer[]1[repl.\(_0\)]{!P \redl{\alpha} P'}
			\end{prooftree}
			\hfill
			\begin{prooftree}
				\hypo{P \redl{\alpha} P'}
				\infer[]1[iter.]{P^* \redl{\alpha} P' ; P^*}
			\end{prooftree}
			\\[.75em]
			{\centering \tikz[baseline=-0.5ex]\draw [thick,dash dot] (0,0)--(8,0);}
			\\[.75em]
			\hfill
			\begin{prooftree}
				\hypo{P \redl{\alpha}  P'}
				\infer[]1[repl.\(_1\)]{!P \redl{\alpha} !P \mid P'}
			\end{prooftree}
			\hfill
			\begin{prooftree}
				\hypo{P \redl{\lambda} P'}
				\hypo{P \redl{\out{\lambda}}  P''}
				\infer[]2[repl.\(_2\)]{!P \redl{\tau} !P \mid (P' \mid P'')}
			\end{prooftree}
			\hfill~
		\end{tcolorbox}
		\caption{Rules of the labeled transition system (LTS) for \ccs, \ccsd, \ccsr and \ccsi}
		\label{fig:ltsrules}
	\end{figure}
	
	\subsubsection{Infinite \ccs}
	Usually, the following three alternative approaches to the addition of infinite behaviors to \ccs are explored and compared~\cite{Busi2003,Busi2009}.
	
	\begin{definition}
		\begin{description}
			\item[\ccsd] (\enquote{\ccs with recursion}) is obtained by adding a denumerable set of constants, ranged over by \(D\), the production rule \(P \coloneqq D\) to \autoref{def:operators}, and the rule const. to the LTS.
			\item[\ccsr] (\enquote{\ccs with replication}) is obtained by adding the production rule \(P \coloneqq !P\) to \autoref{def:operators}, and the rule repl.\(_0\) to the LTS.
			\item[\ccsi] (\enquote{\ccs with iteration}) is obtained by adding the production rules \(P \coloneqq P^*\) and \(P \coloneqq P ; Q\) to \autoref{def:operators}, and the rule iter. to the LTS.
		\end{description}
	\end{definition}
	The complete definition of \ccsi is more involved, as the sequential composition \(P ; Q\) requires its own rules and an auxiliary predicate~\cite[pp. 1195--1196]{Busi2009}.
	
	\begin{remark}
		Adopting the rules repl.\(_1\) and repl.\(_2\) instead of repl.\(_0\) does not change the computational nor decisional power, but gives a finitely branching transition system~\cite[Section 4.3.1]{Busi2009}, among other benefits~\cite[pp.~42-43]{Sangiorgi2001}.
		We will use those rules instead of repl.\(_0\) in \ccsr, as commonly done~\cite{Busi2009,Degano2003}.
	\end{remark}
	
	\subsection{\ccsk: A \enquote{Keyed} Reversible Concurrent Calculus}
	\ccsk captures uncontrolled reversibility using two symmetric LTS--one for forward computation, one for backward computation--that manipulates \emph{keys} marking executed prefixes, to guarantee that reverting  synchronizations cannot be done without both parties agreeing.
	We use the syntax of the latest paper on the topic~\cite{Lanese2021}, that slightly differs~\cite[Remark 4.2]{Lanese2021} with the classical definition~\cite{Phillips2006}.
	However, those changes have no impact since we refrain from using \ccsk's newly introduced structural congruence, as discussed in \autoref{rem:congruence}, p.~\pageref{rem:congruence}.
	
	\begin{definition}[Keys, prefixes and \ccsk processes]
		Let \(\keys=\{m,n,\dots\}\) be a set of \emph{keys}, we let \(k\) range over them.
		Prefixes %
		are of the form \(\alpha[k]\)--we call them \emph{keyed labels}--or \(\alpha\).
		\ccsk processes are \ccs processes where the prefix can also be of the form \(\alpha[k]\), we let \(X\), \(Y\) range over them.
	\end{definition}
	
	The forward LTS for \ccsk, that we denote \(\fwlts{k}{\alpha}\), is given in \autoref{fig:ltsrulesccskfw}--with \(\key\) and \(\std\) defined below.
	The reverse LTS \(\bwlts{k}{\alpha}\) is the exact symmetric of \(\fwlts{k}{\alpha}\)~\cite[Figure 2]{Lanese2021} (it can also be read from~\autoref{fig:provedltsrulesccskfw}), and we write \(X \fbwlts{k}{\alpha} Y\) if \(X \bwlts{k}{\alpha} Y\) or \(X \fwlts{k}{\alpha} Y\).
	For all three types of arrows, we sometimes omit the label and keys when they are not relevant, and mark with \(^*\) their transitive closures.
	As usual, we restrict ourselves to reachable processes, defined below.
		
	\begin{definition}[Standard and reachable processes]
		Given a \ccsk process \(X\), \(\key(X)\) is the set of keys occuring in \(X\), and \(X\) is \emph{standard}--\(\std(X)\)--iff \(\key(X) = \emptyset\).
		If there exists a process \(\orig{X}\) \st \(\std(\orig{X})\) and \(\orig{X} \sllinearrow X\), then \(X\) is \emph{reachable}.
	\end{definition}
	
	The reader eager to see this system in action can fast-forward to \autoref{ex1}, p.~\pageref{ex1}, but should be aware that this example uses proved labels, introduced next.

	\begin{figure}[h]
		\begin{tcolorbox}[title = {Action, Prefix and Restriction}]
			\begin{prooftree}
				\hypo{}
				\infer[left label={\(\std(X)\)}]1[act.]{\alpha. X \fwlts{k}{\alpha}  \alpha[k].X}
			\end{prooftree}
			\hfill
			\begin{prooftree}
				\hypo{X \fwlts{k}{\beta} X'}
				\infer[left label={\(k \neq k'\)}]1[pre.]{\alpha[k']. X \fwlts{k}{\beta} \alpha[k'].X'}
			\end{prooftree}
			
			\begin{prooftree}
				\hypo{ X   \fwlts{k}{\alpha}X '}
				\infer[left label={\(a \notin \{\alpha, \out{\alpha}\}\)}]1[res.]{X  \bs a  \fwlts{k}{\alpha}X ' \bs a}
			\end{prooftree}
		\end{tcolorbox}
		
		\begin{tcolorbox}[adjusted title=Parallel Group]
			\begin{prooftree}
				\hypo{X \fwlts{k}{\alpha}X'}
				\infer[left label={\(k \notin \key(Y)\)}]
				1[\(\mid_{\L}\)]{X \mid  Y \fwlts{k}{\alpha} X' \mid Y}
			\end{prooftree}
			\hfill
			\begin{prooftree}
				\hypo{Y \fwlts{k}{\alpha}Y'}
				\infer[left label={\(k \notin \key(X)\)}]
				1[\(\mid_{\R}\)]{X \mid  Y  \fwlts{k}{\alpha}X \mid Y'}
			\end{prooftree}
			\\[1.4em]
			\begin{prooftree}
				\hypo{X \fwlts{k}{\lambda}  X'}
				\hypo{Y \fwlts{k}{\out{\lambda}}  Y'}
				\infer%
				2[syn.]{X \mid Y \fwlts{k}{\tau}  X' \mid Y'}
			\end{prooftree}
		\end{tcolorbox}
		\begin{tcolorbox}[adjusted title=Sum Group]
			\hfill
			\begin{prooftree}
				\hypo{X \fwlts{k}{\alpha} X'}
				\infer[left label={\(\std(Y)\)}]1[\(+_{\L}\)]{X + Y \fwlts{k}{\alpha} X' + Y}
			\end{prooftree}
			\hfill
			\begin{prooftree}
				\hypo{Y \fwlts{k}{\alpha} Y'}
				\infer[left label={\(\std(X)\)}]1[\(+_{\R}\)]{X + Y \fwlts{k}{\alpha} X + Y'}
			\end{prooftree}
			\hfill~
		\end{tcolorbox}
		\caption{Rules of the forward labeled transition system (LTS) for \ccsk}
		\label{fig:ltsrulesccskfw}
	\end{figure}
	
	\section{A New Causal Semantics for \ccsk}
	\label{sec:causal-sem}	
	
	As discussed in the introduction, the only causal semantics for \ccsr we are aware of~\cite{Degano2003} remained unnoticed, despite some interesting qualities: \begin{enumerate*}
		\item it enables the definition of causality for replication while agreeing with pre-existing causal semantics of \ccs and \ccsd~\cite[Theorem 1]{Degano2003}
		\item it leverages the technique of \emph{proved} transition systems that encodes information about the derivation in the labels~\cite{Degano2001},
		\item it was instrumental in one of the first result connecting implicit computational complexity and distributed processes~\cite{Demangeon2018},
		\item last but not least, as we will see below, it allows to define an elegant notion of causality for \ccsk with \enquote{built-in} reversibility, as \emph{the exact same definition will be used for forward and backward transitions}, without making explicit mentions of the keys or directions.
	\end{enumerate*}	
	We comment on existing causal semantics for \ccs-like reversible process calculi before diving into the technical details.

\textbf{In \rccs,} the definition of concurrency fluctuated between a condition on memory inclusion for composable transitions~\cite[Definition 3.11]{Krivine2006} and a condition on disjointness of memories on coinitial transitions~\cite[Definition 7]{Danos2004}, both requiring the entire memory of the thread to label the transitions.
		Concurrency have also been studied using configuration structures for \rccs semantics~\cite{Aubert2015d,Aubert2016jlamp,Aubert2020b}.
		We conjecture (an adaptation of) our definition (to \rccs) to be equivalent to those three definitions, that have not been proven themselves equivalent.
		In any case, whether those definitions conservatively extend \ccs's causality have not been proven, and they do not account for infinite behaviors.
	
	\textbf{In \ccsk,} forward and reverse diamond properties were proven using conditions on keys and \enquote{joinable} transitions \cite[Propositions 5.10 and 5.19]{Phillips2006}, but to our knowledge no \enquote{definitive} definition of concurrency was proposed.
		Notions of concurrency could probably be \enquote{imported} from semantical models of \ccsk~\cite{Graversen2021,Phillips2007,Ulidowski2014}, but have not been to our knowledge.
	
	Those takes on concurrency impacted all the extensions to \rccs and \ccsk, \eg reversible higher-order \(\pi\)-calculus~\cite[Definition 9]{Lanese2016}, reversible \(\pi\)-calculus~\cite[Definition 4.1]{Cristescu2013} or croll-\(\pi\)~\cite[Definition 1]{Lanese2013}, where similar definitions of concurrency were adopted.
	One of the downside of those syntactical definitions is, in our opinion, that too much importance is given to keys or identifiers, that should only be technical annotations disallowing processes that have been synchronized to backtrack independently.
	We have previously defended that identifier should be considered only up to isomorphisms~\cite[p.~13]{Aubert2020b}, or explicitly generated by a built-in mechanism~\cite[p.~152]{Aubert2021d}, and re-inforce this point of view by providing \ccsk with a concurrency mechanism that does not need to consider keys.
		
	We believe our choice is additionally compact, elegant and suited for reversible computation: defining concurrency on composable transitions first allows \emph{not} to consider keys or identifiers in our definition, as the LTS guarantees that the same key will not be re-used.
	\emph{Then}, the loop lemma allows to \enquote{reverse} transitions so that concurrency on coinitial transitions can be defined from concurrency on composable transitions.
	This allows to carry little information in the labels--the direction is not needed--and to have a definition insensible to keys and identifiers for the very modest cost of prefixing labels with some annotation tracing back the thread(s) performing the transition.
	
	\subsection{Proved Labeled Transition System for \ccsk}
	\label{ssec:proved-ccsk}
	
	We adapt the proved transition system~\cite{Carabetta1998,Degano2003,Degano1992} to \ccsk: this technique enriches the transitions label with prefixes that describe parts of their derivation, to keep track of their dependencies or lack thereof.
	We adapt an earlier formalism~\cite{Degano1999} that does not record information about sums~\cite[footnote 2]{Degano2003}, %
	but extend the concurency relation to internal (\ie \(\tau\)-) transitions, omitted from recent work~\cite[Definition 3]{Degano2003} but present in older articles~\cite[Definition 2.3]{Carabetta1998}.
	
	\begin{definition}[Enhanced keyed labels]
		\label{def:enhanced-keyed-labels}
		Let \(\upsilon\) range over strings in \(\{\mid_{\L}, \mid_{\R}%
		\}^*\), \emph{enhanced keyed labels} are defined as
		\[ \theta \coloneqq \upsilon \alpha[k]  ~\|~ \upsilon \langle \mid_{\L} \upsilon_{\L} \alpha[k], \mid_{\R} \upsilon_{\R} \out{\alpha}[k]\rangle \]
		We write \(\enhanced\) the set of enhanced keyed labels, and define \(\ell : \enhanced \to \labels\) and \(\kay: \enhanced \to \keys\): %
		\begin{align*}
			\ell(\upsilon \alpha[k]) & = \alpha &   &   & \ell(\upsilon \langle \mid_{\L} \upsilon_{\L} \alpha[k], \mid_{\R} \upsilon_{\R} \out{\alpha}[k]\rangle) & = \tau \\
			\kay(\upsilon \alpha[k]) & = k      &   &   & \kay(\upsilon \langle \mid_{\L} \upsilon_{\L} \alpha[k], \mid_{\R} \upsilon_{\R} \out{\alpha}[k]\rangle) & = k%
		\end{align*}
	\end{definition}
	
	We present in~\autoref{fig:provedltsrulesccskfw} the rules for the \emph{proved} forward and backward LTS for \ccsk.
	The rules \strut{\(\mid_{\R}\)}, \strut{\(\Rev{\mid_{\R}}\)}, \strut{\(+_{\R}\)} and  \strut{\(\Rev{+_{\R}}\)} are omitted but can easily be inferred.
	
	\begin{figure}[ht]
		\begin{tcolorbox}[title = {Action, Prefix and Restriction}, sidebyside]
			\begin{tcolorbox}[adjusted title=Forward]
				\begin{prooftree}
					\hypo{}
					\infer[left label={\(\std(X)\)}]1[act.]{\alpha. X \fwlts{k}{\alpha}  \alpha[k].X}
				\end{prooftree}
				\\[.9em]
				\begin{prooftree}
					\hypo{X \redl{\theta} X'}
					\infer[left label={\(\kay(\theta) \neq k\)}]1[pre.]{\alpha[k]. X \redl{\theta} \alpha[k].X'}
				\end{prooftree}
				\\[.9em]
				\begin{prooftree}
					\hypo{ X \redl{\theta} X '}
					\infer[left label={\(\ell(\theta) \notin \{\alpha, \out{\alpha}\}\)}]1[res.]{X  \bs a  \redl{\theta} X ' \bs a}
				\end{prooftree}
			\end{tcolorbox}
			\tcblower
			\begin{tcolorbox}[adjusted title=Backward]
				\begin{prooftree}
					\hypo{}
					\infer[left label={\(\std(X)\)}]1[\Rev{\text{act.}}]{ \alpha[k].X \bwlts{k}{\alpha} \alpha. X}
				\end{prooftree}
				\\[.9em]
				\begin{prooftree}
					\hypo{X' \revredl{\theta} X}
					\infer[left label={\(\kay(\theta) \neq k\)}]1[\Rev{\text{pre}.}]{\alpha[k].X'\revredl{\theta} \alpha[k]. X }
				\end{prooftree}
				\\[.9em]
				\begin{prooftree}
					\hypo{ X' \revredl{\theta} X}
					\infer[left label={\(\ell(\theta) \notin \{\alpha, \out{\alpha}\}\)}]1[\Rev{\text{res.}}]{X ' \bs a\revredl{\theta} X  \bs a  }
				\end{prooftree}
			\end{tcolorbox}
		\end{tcolorbox}
		\begin{tcolorbox}[title = Parallel Group, sidebyside]
			\begin{tcolorbox}[adjusted title=Forward]
				\begin{prooftree}
					\hypo{X \redl{\theta} X'}
					\infer[left label={\(\kay(\theta) \notin \key(Y)\)}]
					1[\(\mid_{\L}\)]{X \mid  Y \redl{\mid_{\L}\theta} X' \mid Y}
				\end{prooftree}
				\\[.9em]
				\begin{prooftree}
					\hypo{X \fwlts{k}{\theta_{\L} \lambda}  X'}
					\hypo{Y \fwlts{k}{\theta_{\R} \out{\lambda}}  Y'}
					\infer2[syn.]{X \mid Y \redl{\langle \mid_{\L} \theta_{\L} \lambda \protect{[k]}, \mid_{\R} \theta_{\R} \out{\lambda} \protect{[k]} \rangle} X' \mid Y'}
				\end{prooftree}
			\end{tcolorbox}
			\tcblower
			\begin{tcolorbox}[adjusted title=Backward]
				\begin{prooftree}
					\hypo{X' \revredl{\theta} X}
					\infer[left label={\(\kay(\theta) \notin \key(Y)\)}]
					1[\Rev{\mid_{\L}}]{X' \mid  Y \revredl{\mid_{\L}\theta} X \mid Y}
				\end{prooftree}
				\\[.9em]
				\begin{prooftree}
					\hypo{X' \bwlts{k}{\theta_{\L} \lambda}  X}
					\hypo{Y' \bwlts{k}{\theta_{\R} \out{\lambda}}  Y}
					\infer%
					2[\Rev{\text{syn.}}]{X' \mid Y' \revredl{\langle \mid_{\L} \theta_{\L} \lambda \protect{[k]}, \mid_{\R} \theta_{\R} \out{\lambda} \protect{[k]} \rangle} X \mid Y}
				\end{prooftree}
			\end{tcolorbox}
		\end{tcolorbox}
		
		\begin{tcolorbox}[adjusted title=Sum Group, sidebyside]
			\begin{tcolorbox}[adjusted title=Forward]
				\begin{prooftree}
					\hypo{X \redl{\theta} X'}
					\infer[left label={\(\std(Y)\)}]1[\(+_{\L}\)]{X + Y \redl{%
							\theta} X' + Y}
				\end{prooftree}
			\end{tcolorbox}
			\tcblower
			\begin{tcolorbox}[adjusted title=Backward]
				\begin{prooftree}
					\hypo{X' \revredl{\theta} X}
					\infer[left label={\(\std(Y)\)}]1[\(\Rev{+_{\L}}\)]{X' + Y \revredl{%
							\theta} X + Y}
				\end{prooftree}
			\end{tcolorbox}
		\end{tcolorbox}
		\caption{Rules of the \emph{proved} LTS for \ccsk}
		\label{fig:provedltsrulesccskfw}
	\end{figure}

	\begin{definition}[Dependency relation]%
	\label{def:deprel}
	The \emph{dependency relation} on enhanced keyed labels is induced by the axioms below, for \(d \in \{\L, \R\}\):
	\begin{align*}
	\alpha[k] & \dep \theta &&& 
	\mid_{d}\theta & \dep \mid_{d} \theta' && \text{if \(\theta \dep \theta'\)} &&&
	\langle \theta_{\L}, \theta_{\R} \rangle & \dep \theta && \text{if \(\exists d \text{ \st} \theta_d \dep \theta\)}%
	\\
	&&&&   \theta & \dep \langle \theta_{\L}, \theta_{\R} \rangle &&\text{if \(\exists d \text{ \st} \theta \dep \theta_d\)}
	&&&  \langle \theta_{\L}, \theta_{\R} \rangle &\dep \langle \theta'_{\L}, \theta'_{\R} \rangle && \text{if \(\exists d \text{ \st} \theta_d \dep \theta'_d\)}
	\end{align*}
\end{definition}
	A dependency \(\theta_0 \dep \theta_1\) means \enquote{whenever there is a trace in which \(\theta_0\) occurs before \(\theta_1\), then the two associated transitions are causally related}.
	The following definitions will enable more formal examples, but we can stress that
	\begin{enumerate*}
		\item the first rule enforces the fact that executing or reversing a prefix at top level, \eg \(\alpha.X \fwlts{k}{\alpha} \alpha[k].X\) or \(\alpha[k].X \bwlts{k}{\alpha} \alpha.X\), makes the prefix (\(\alpha[k]\)) a dependency of all further transitions;
		\item as the forward and backward versions of the same rule share the same enhanced keyed labels,
		a trace where a transition and its reverse both occur will have the first occurring be a dependency of the second, as \(\dep\) is reflexive;
		\item whenever reasoning by case over \(\theta\), we can assume \wlg that \eg an enhanced key label \(\mid_{\R} \theta\) was \enquote{ultimately} produced by a transition from a process of the form \(X \mid Y\): while the actual process may have \(+\), \(\bs a\) and \(\alpha[k]\) operators \enquote{wrapping} \(X \mid Y\), they have no impact on the dependencies nor on the properties we will be proving, and thus we can assume   \wlg that a term capable of performing a transition labeled by a prefix (\resp prefixed by \(\mid_{\L}\), \(\mid_{\R}\), \(\langle \theta_{\L}, \theta_{\R}\rangle\)) can always be assumed to have for primary connector the same prefix (\resp a parallel composition). This is made formal in \autoref{app:sub-term}, for the calculus including the replication operator as well, and formally used only in \autoref{lem:bwconc}.
	\end{enumerate*}
	
	\begin{definition}[Transitions and traces]
		In a \emph{transition} \(t: X\frevredl{\theta} X'\), \(X\) is the \emph{source}, and \(X'\) is the \emph{target} of \(t\). %
		Two transitions are \emph{coinitial} (\resp \emph{cofinal}) if they have the same source (\resp  target).
		Transitions \(t_1\) and \(t_2\) are \emph{composable}, \(t_1;t_2\),  if the target of \(t_1\) is the source of \(t_2\).
		The \emph{reverse} of \(t : X' \revredl{\theta} X'\) is \(\Rev{t} : X \redl{\theta} X'\), and similarly if \(t\) is backward, letting \(\Rev{(\Rev{t})} = t\)\footnote{The existence and uniqueness of the reverse transition is immediate in \ccsk. This property, known as the loop lemma (\autoref{lemma:loop}) is sometimes harder to obtain.}.
		
		A sequence of pairwise composable transitions \(t_1; \cdots; t_n\) is called a \emph{trace}, denoted \(T\), and \(\epsilon\) is the empty trace.
		\end{definition}
	
	\begin{definition}[Causality relation]
		\label{def:causal-rel}
		Let \(T\) be a trace \(X_1 \frevredl{\theta_1} \cdots \frevredl{\theta_n} X_n\) and \(i, j \in \{1, \hdots, n\}\) with \(i < j\),  \emph{\(\theta_i\) causes \(\theta_j\) in \(T\)} (\(\theta_i \dep_{T} \theta_j\)) iff \(\theta_i \dep \theta_j\).
	\end{definition}
	
	\begin{definition}[Concurrency]
		Let \(T\) be a trace \(X_1 \frevredl{\theta_1} \cdots \frevredl{\theta_n} X_n\) and \(i, j \in \{1, \hdots, n\}\), %
		 \emph{\(\theta_i\) is concurrent with \(\theta_j\)} (\(\theta_i \conc_T \theta_j\), or simply \(\theta_i \conc \theta_j\)) iff neither 
		\(\theta_i \dep_{T} \theta_j\) nor \(\theta_j \dep_{T} \theta_i\).
	\end{definition}
	
	\begin{example}
		\label{ex1}
		Consider the following trace, dependencies, and concurrent transitions, where the subscripts to \(\dep\) and \(\conc\) have been omitted:
		
		\noindent
		\begin{minipage}{0.42\textwidth}%
			\begin{align*}
			& (a . \out{b}) \mid (b + c)                                                                             \\
			& \redl{\protect{\mid_{\L}a[m]}} a[m] . \out{b} \mid b +c                                                \\
			& \redl{\protect{\mid_{\L}\out{b}[n]}} a[m] . \out{b}[n] \mid b +c                                       \\
			& \redl{\protect{\mid_{\R}c[n']}} a[m] . \out{b}[n] \mid b + c[n']                                      \\
			& \revredl{\protect{\mid_{\L}\out{b}[n]}} a[m] . \out{b} \mid b + c[n']                                  \\
			& \revredl{\protect{\mid_{\R}c[n']}} a[m] . \out{b} \mid b +c                                            \\
			& \redl{\protect{\langle \mid_{\L}\out{b}[n], \mid_{\R}b[n]\rangle}} a[m] . \out{b}[n] \mid b[n] + c 
			\end{align*}
		\end{minipage}%
		\hfill
		\begin{minipage}{0.54\textwidth}%
			\begin{align*}
			\shortintertext{And we have, \eg}
			& \mid_{\L} \out{b}[n] \dep \mid_{\L}\out{b}[n] &   & \text{ as } \out{b}[n] \dep \out{b}[n] \\
			& \mid_{\L} a[m] \dep \mid_{\L} \out{b}[n]      &   & \text{ as } a[m] \dep \out{b}[n]       
			\shortintertext{and, as a consequence,}
			& \mid_{\L}a[m] \dep \langle \mid_{\L}\out{b}[n], \mid_{\R}b[n]\rangle\\
			\shortintertext{but}
			& \mid_{\L} \out{b}[n] \conc \mid_{\R} c[n]
			\shortintertext{as labels prefixed by \(\mid_{\L}\) and \(\mid_{\R}\) are never causes of each others.}
			\end{align*}
		\end{minipage}
	\end{example}
	
	To prove the results in the next section, we need a convoluted but straightforward lemma that decomposes concurrent traces involving two threads into one trace involving one thread while maintaining concurrency, \ie proving that transitions \eg of the form \(T: X \mid Y \frevredl{\mid_{\L}\theta} X' \mid Y \frevredl{\mid_{\L}\theta'} X'' \mid Y\) with \(\mid_{\L}\theta \conc_{T} \mid_{\L}\theta'\) can be decomposed into transitions \(T': X \frevredl{\theta} X' \frevredl{\theta'} X''\) with \(\theta \conc_{T'} \theta'\).
	
	\begin{restatable}[Decomposing concurrent transitions]{lemma}{lempartoftrans}
		\label{lem:part_of_trans}
		Let \(i \in \{1, 2\}\) and \(\theta_i \in \{\mid_{\L} \theta_i', \mid_{\R} \theta_i'', \langle \mid_{\L} \theta_i', \mid_{\R} \theta_i''\rangle\}\), and define \(\pi_{\L} (X_{\L} \mid X_{\R}) = X_{\L}\), \(\pi_{\L} (\mid_{\L} \theta) = \theta\), \(\pi_{\L}(\langle \mid_{\L} \theta_{\L}, \mid_{\R} \theta_{\R}\rangle) = \theta_{\L}\), \(\pi_{\L}(\mid_{\R} \theta) = \text{undefined}\), and define similarly \(\pi_{\R}\).
		
		Whenever \(T: X_{\L} \mid X_{\R} \frevredl{\theta_1} Y_{\L} \mid Y_{\R} \frevredl{\theta_2} Z_{\L} \mid Z_{\R}\) with \(\theta_1 \conc_{T} \theta_2\), then for \(d \in \{\L, \R\}\), if \(\pi_{d}(\theta_1)\) and \(\pi_{d}(\theta_2)\) are both defined, then, \(\pi_{d}(\theta_1) \conc_{\pi_{d}(T)} \pi_{d}(\theta_2)\) with \(\pi_{d}(T) : \pi_{d}(X_{\L} \mid X_{\R}) \frevredl{\pi_{d}(\theta_1)} \pi_{d}(Y_{\L} \mid Y_{\R}) \frevredl{\pi_{d}(\theta_2)}\pi_{d}(Z_{\L} \mid Z_{\R})\).
		
	\end{restatable}
	
	\begin{proof}
		The trace \(\pi_d(T)\) exists by virtue of the rule \(\mid_{d}\), syn.\ or their reverses.
		That \(\pi_{d}(\theta_1) \conc_{\pi_{d}(T)} \pi_{d}(\theta_2)\) essentially comes from the fact than if one was causing the other, then \(\theta_1 \conc_{T} \theta_2\) would not hold, a contradiction.
	\end{proof}

	\subsection{Diamonds and Squares: Concurrency in Action}
	\label{sec:diamond}
	Square properties and concurrency diamonds express that concurrent transitions are \emph{actually} independent, in the sense that they can \enquote{later on} agree (the square property, \autoref{thm:sp}) or be swapped (the forward diamond, \autoref{thm:forward}, if both transitions are forward, or \enquote{sideways} diamond, \autoref{thm:side}, if they have opposite directions).
	That our definition of concurrency enables those, \emph{and} allows to inter-prove them, is a good indication that it is resilient and convenient to use.

	\begin{restatable}[Forward diamond]{theorem}{thmfwdiamond}\label{thm:forward}
		For all \(X \redl{\theta_1} X_1 \redl{\theta_2} Y\) with \(\theta_1 \conc \theta_2\), there exists \(X_2\) \st  \(X \redl{\theta_2} X_2 \redl{\theta_1} Y\).
	\end{restatable}
	
	The proof, p.~\pageref{proof:thm:forw}, requires a particular care when \(X\) is not standard.
	Using pre.\ is transparent from the perspective of enhanced keyed labels, as no \enquote{memory} of its usage is stored in the label of the transition.
	This is essentially because--exactly like for act.\--all the dependency information is already present in the term or its enhanced keyed label.
	To make this more formal, we introduce a function that \enquote{removes} a keyed label, and prove that it does not affect derivability. %
	
	\begin{definition}
		Given \(\alpha\) and \(k\), we define \(\rem_{\alpha[k]}\) by \(\rem_{\alpha[k]} (0)  = 0\) and 
		\begin{align*}
			\rem_{\alpha[k]} (\beta.X) & = \beta.X                   &   &   & \rem_{\alpha[k]}(X \mid Y) & = \rem_{\alpha[k]}(X) \mid \rem_{\alpha[k]}(Y) \\
			\rem_{\alpha[k]} (X\bs a)  & = 	(\rem_{\alpha[k]}X)\bs a &   &   & \rem_{\alpha[k]}(X + Y)    & = \rem_{\alpha[k]}(X) + \rem_{\alpha[k]}(Y)    \\
			\shortintertext{\centering \(\rem_{\alpha[k]} (\beta[m].X) =\begin{dcases*} X & if \(\alpha = \beta\) and \(k = m\) \\ \beta[m].\rem_{\alpha[k]}(X) & otherwise \end{dcases*}\)}
	\end{align*}
	We let \(\rem_{k}^{\lambda} = \rem_{\lambda[k]} \circ \rem_{\out{\lambda}[k]}\) if \(\lambda \in \labels \bs \{\tau\}\), \(\rem_{k}^{\tau} = \rem_{\tau[k]}\) otherwise.
\end{definition}

The function \(\rem_{\alpha[k]}\) simply looks for an occurrence of \(\alpha[k]\) and removes it: as there is at most one, there is no need for a recursive call when it is found.
This function preserves derivability of transitions that do not involve the key removed:

\begin{lemma}
	\label{lem:rm}
	For all \(X\), \(\alpha\) and \(k\), \(X \frevredl{\theta} Y\) with \(\kay(\theta) \neq k\) iff \(\rem_{k}^{\alpha}(X) \frevredl{\theta} \rem_{k}^{\alpha}(Y)\).	
\end{lemma}

\begin{proof}
	Assume \(\alpha[k]\) or \(\out{\alpha}[k]\) (if \(\alpha \neq \tau\)) occur in \(X\) (otherwise the result is straightforward), as \(\kay(\theta) \neq k\), the same holds for \(Y\).
	As keys occur at most twice, attached to complementary names, in reachable processes~\cite[Lemma 3.4]{Lanese2021}, \(k \notin \key(\rem_{k}^{\alpha}(X)) \cup \key(\rem_{k}^{\alpha}(Y))\).
	Then the proof follows from simple induction on the length of the derivation for \(X \frevredl{\theta} Y\): as neither pre.\ nor \Rev{\text{pre.}}\ change the enhanced keyed label, we can simply \enquote{take out} the occurrences of those rules when they concern \(\alpha[k]\) or \(\out{\alpha}[k]\) and still obtain a valid derivation, with the same enhanced keyed label, hence yielding \(\rem_{k}^{\alpha}(X) \frevredl{\theta} \rem_{k}^{\alpha}(Y)\).
	For the converse direction, pre.\ or \Rev{\text{pre.}} can be re-added to the derivation tree and in the same location, as \(k\) is fresh in  \(\rem_{k}^{\alpha}(X)\) and \(\rem_{k}^{\alpha}(Y)\).
\end{proof}

\begin{proof}[of \protect{\autoref{thm:forward}}]
	\label{proof:thm:forw}
	The proof proceeds by induction on the length of the deduction for the derivation for \(X_1 \redl{\theta_1} X_2\)%
	, using \autoref{lem:part_of_trans} to enable the induction hypothesis if \(\theta_1 = \mid_{\L}\theta_1'\), \(\mid_{\R}\theta_1'\) or \(\langle \mid_{\L}\theta_{\L}', \mid_{\R}\theta_{\R}'\rangle\).
	The only delicate case is if the last rule is pre.: in this case, there exists \(\alpha\), \(k\), \(X'\) and \(X_1'\) \st  \(X = \alpha[k].X' \redl{\theta_1} \alpha[k].X_1' = X_1\) and \(\kay(\theta_1) \neq k\).
	As \(\alpha[k].X_1' \redl{\theta_2} Y\), \(\kay(\theta_2) \neq k\)~\cite[Lemma 3.4]{Lanese2021}, and since \(\theta_1 \conc \theta_2\), we apply \autoref{lem:rm} twice to obtain the trace \(T\):
	\[\rem^{\alpha}_k (\alpha[k].X') = X' \redl{\theta_1} \rem^{\alpha}_k (\alpha[k].X'_1) = X'_1 \redl{\theta_2} \rem^{\alpha}_k (Y)\]
	with \(\theta_1 \conc_{T} \theta_2\), and we can use the induction hypothesis to obtain \(X_2\) \st  \(X' \redl{\theta_2} X_2 \redl{\theta_1} \rem^{\alpha}_k (Y)\).
	Since \(\kay(\theta_2) \neq k\), we can append pre.\ to the derivation of \(X' \redl{\theta_2} X_2\) to obtain \(\alpha[k] . X' = X \redl{\theta_2} \alpha[k]. X_2\).
	Using \autoref{lem:rm} one last time, we obtain that \(\rem^{\alpha}_k(\alpha[k].X_2) = X_2 \redl{\theta_1} \rem^{\alpha}_k(Y)\) implies \(\alpha[k].X_2 \redl{\theta_1} Y\), which concludes this case.
\end{proof}

\begin{example}
	Re-using \autoref{ex1}, since \(\mid_{\L} \out{b}[n] \conc \mid_{\R}c[n']\) in 	
	\begin{align*}
	a[m] . \out{b} \mid b +c & \redl{\protect{\mid_{\L}\out{b}[n]}} a[m] . \out{b}[n] \mid b +c  \redl{\protect{\mid_{\R}c[n']}} a[m] . \out{b}[n] \mid b +c[n']\text{,}\\
	\shortintertext{\autoref{thm:forward} allows to re-arrange this trace as}
	a[m] . \out{b} \mid b +c & \redl{\protect{\mid_{\R}c[n']}} a[m] . \out{b} \mid b + c[n']  \redl{\protect{\mid_{\L}\out{b}[n]}} a[m] . \out{b}[n] \mid b +c[n']\text{.}
	\end{align*}
\end{example}

\begin{restatable}[Sideways diamond]{theorem}{thmsidediamond}\label{thm:side}
	For all \(X \redl{\theta_1} X_1 \revredl{\theta_2} Y\) with \(\theta_1 \conc \theta_2\), there exists \(X_2\) \st  \(X \revredl{\theta_2} X_2 \redl{\theta_1} Y\).
\end{restatable}

It should be noted that in the particular case of \(t;\Rev{t}: X \redl{\theta_1} X_1 \revredl{\theta_1} X\), \(\theta_1 \dep \theta_1\) by reflexivity of \(\dep\) and hence the sideways diamond cannot apply.

\begin{proof}
	We can re-use the proof of \autoref{thm:forward} almost as it is: Lemmas \ref{lem:part_of_trans} and  \ref{lem:rm} hold for both directions, so we can use them for backward transitions.	
	The only case that slightly diverges is if the deduction for \(X \redl{\theta_1} X_1\) have for last rule pre. 
	In this case, \(\alpha[k].X' \redl{\theta_1} \alpha[k].X'_1 \revredl{\theta_2} Y\), but we cannot deduce that \(\kay(\theta_2) \neq k\) immediately.
	However, if \(\kay(\theta_2) = k \), then we would have \(\alpha[k].X_1' \revredl{\protect{\alpha[k]}} \alpha.Y\), but this application of \Rev{\text{pre.}} is not valid, as \(\std(X_1')\) does not hold, since \(X_1'\) was obtained from \(X'\) after it made a \emph{forward} transition.
	Hence, we obtain that  \(\key(\theta_2) \neq k\) and we can carry out the rest of the proof as before.
\end{proof}

\begin{example}
	Re-using \autoref{ex1}, since \(\mid_{\R}c[n'] \conc \mid_{\L}\out{b}[n]\) in 	
	\begin{align*}
	a[m] . \out{b}[n] \mid b +c & \redl{\protect{\mid_{\R}c[n']}} a[m] . \out{b}[n] \mid b + c[n'] \revredl{\protect{\mid_{\L}\out{b}[n]}} a[m] . \out{b} \mid b + c[n'] \text{,}\\
	\shortintertext{\autoref{thm:side} allows to re-arrange this trace as}
	a[m] . \out{b}[n] \mid b + c &\revredl{\protect{\mid_{\L}\out{b}[n]}} a[m] . \out{b} \mid b +c \redl{\protect{\mid_{\R}c[n']}} a[m] . \out{b} \mid b + c[n']\text{.}
	\end{align*}
\end{example}

Concurrency on coinitial traces is defined using concurrency on composable traces and the loop lemma, immediate in \ccsk. %

\begin{lemma}[Loop lemma~{\cite[Prop. 5.1]{Phillips2006}}]
	\label{lemma:loop}
	For all \(t: X \redl{\theta} X'\), there exists a unique \(\Rev{t}: X' \revredl{\theta} X\), and reciprocally.
\end{lemma}

\begin{definition}[Coinitial concurrency]
	\label{def:co-init-conc}
	Let \(t_1: X \frevredl{\theta_1} Y_1\) and \(t_2: X \frevredl{\theta_2} Y_2\) be two coinitial transitions, \(\theta_1\) is concurrent with \(\theta_2\) (\(\theta_1 \conc \theta_2\)) iff \(\theta_1 \conc \theta_2\) in the trace \(\Rev{t_1} ; t_2 : Y_1 \frevredl{\theta_1} X \frevredl{\theta_2} Y_2\). 
\end{definition}

\begin{theorem}[Square property]
	\label{thm:sp}
	For all \(t_1 :X\frevredl{\theta_1}X_1\) and \(t_2:X\frevredl{\theta_2}X_2\) with \(\theta_1 \conc \theta_2\), there exist \(t'_1 :X_1\frevredl{\theta_2}Y\) and \(t'_2 :X_2\frevredl{\theta_1}Y\).
\end{theorem}

\begin{proof}
	Since \(t_1\) and \(t_2\) are concurrent, by \autoref{def:co-init-conc} we have that \(\theta_1 \conc \theta_2\) in \(\Rev{t_1};t_2 :X_1\frevredl{\theta_1}X \frevredl{\theta_2} X_2\).
	Hence, depending on the direction of the arrows, and possibly using the loop lemma, we obtain by the forward or sideways diamond (Theorems~\ref{thm:forward} and \ref{thm:side}) \(t''_1;t''_2 : X_1  \frevredl{\theta_2} Y \frevredl{\theta_1} X_2\), and we let \(t'_1 = t''_1\) and \(t'_2 = \Rev{{t''}_2}\):

	\begin{tikzpicture}[anchor=base, baseline=-1.2cm]
	\node (X) {\(X\)};
	\node (X1) [below left = .8cm and .4cm of X] {\(X_1\)};	
	\draw [arrows={- angle 45}, ->>] (X)--node[left]{\(\theta_1\)} (X1);
	\node (X2) [below right = .8cm and .4cm of X] {\(X_2\)};	
	\draw [arrows={- angle 45}, ->>] (X)--node[right]{\(\theta_2\)} (X2);
	\end{tikzpicture}
	\(\xRightarrow{\text{Loop}}\)
	\begin{tikzpicture}[anchor=base, baseline]
	\node (X) {\(X\)};
	\node (X1) [above = .8cm of X] {\(X_1\)};	
	\draw [arrows={- angle 45}, ->>] (X1)--node[right]{\(\theta_1\)} (X);
	\node (X2) [below = .8cm of X] {\(X_2\)};	
	\draw [arrows={- angle 45}, ->>] (X)--node[right]{\(\theta_2\)} (X2);
	\end{tikzpicture}
	\(\xRightarrow{\text{Diamonds}}\)
	\begin{tikzpicture}[anchor=base, baseline]
	\node (X) {\(Y\)};
	\node (X1) [above = .8cm of X] {\(X_1\)};	
	\draw [arrows={- angle 45}, ->>] (X1)--node[right]{\(\theta_2\)} (X);
	\node (X2) [below = .8cm of X] {\(X_2\)};	
	\draw [arrows={- angle 45}, ->>] (X)--node[right]{\(\theta_1\)} (X2);
	\end{tikzpicture}
	\(\xRightarrow{\text{Loop}}\)
	\begin{tikzpicture}[anchor=base, baseline=-1.2cm]
	\node (X) {\(X\)};
	\node (X1) [below left = .8cm and .4cm of X] {\(X_1\)};	
	\draw [arrows={- angle 45}, ->>] (X)--node[left]{\(\theta_1\)} (X1);
	\node (X2) [below right = .8cm and .4cm of X] {\(X_2\)};	
	\draw [arrows={- angle 45}, ->>] (X)--node[right]{\(\theta_2\)} (X2);
	\node (Y) [below = 2.2cm of X] {\(Y\)};	
	\draw [arrows={- angle 45}, ->>] (X1)--node[left]{\(\theta_2\)} (Y);
	\draw [arrows={- angle 45}, ->>] (X2)--node[right]{\(\theta_1\)} (Y);
	\end{tikzpicture}	
\end{proof}

\begin{example}
	Following \autoref{ex1}, we can get \eg from \(a[m] . \out{b}[n] \mid b +c \redl{\protect{\mid_{\R}b[n']}} a[m] . \out{b}[n] \mid b[n'] +c \) and \(a[m] . \out{b}[n] \mid b +c  \revredl{\protect{\mid_{\L}\out{b}[n]}} a[m] . \out{b} \mid b +c\) the transitions converging to \(a[m] . \out{b} \mid b[n'] +c\).
\end{example}

\subsection{Sanity Checks: Causal Consistency, Optimality and More}
\label{sec:sanity}

Causality for a semantics of concurrent computations should satisfy a variety of critera, the squares and diamonds only being the starting point generally.
Optimality with respect to the reduction semantics was discussed for the \(\pi\)-calculus~\cite[Theorem 5.6]{Cristescu2013}, and later refined along other criteria~\cite{Cristescu2015b}, for notions of concurrency defined respectively in terms of memory inclusion and rigid families.
For reversible systems in particular, other criteria have emerged, along which causal consistency--arguably the most important--that guarantees that transitions cannot be undone unless their consequences have been too.

\subsubsection{Causal Consistency}

Formally, causal consistency states that any two coinitial and cofinal traces are causally equivalent:

\begin{definition}[Causally equivalence]
	Two traces \(T_1\), \(T_2\) are \emph{causally equivalent}, if they are in the least equivalence relation closed by composition satisfying \(t; \Rev{t} \eq \epsilon\)  and \(t_1;t_2' \eq t_2;t_1'\) for any
	\(t_1 ; t_2' : X \frevredl{\theta_1} \frevredl{\theta_2} Y\), \(t_2 ; t_1' : X \frevredl{\theta_2}\frevredl{\theta_1} Y\). 
\end{definition}

\begin{theorem}[Causal consistency]
	\label{thm:cc}
	All coinitial and cofinal traces are causally equivalent.
\end{theorem}

The \enquote{axiomatic approach} to reversible computation~\cite{Lanese2020} allows to obtain causal consistency from other properties that are generally easier to prove.

\begin{restatable}[Backward transitions are concurrent]{lemma}{bwconc}
	\label{lem:bwconc}
	Any two different coinitial backward transitions \(t_1: X\revredl{\theta_1}X_1\) and \(t_2:X\revredl{\theta_2}X_2\) are concurrent.
\end{restatable}

\begin{proof}
	The proof is by induction on the size of \(\theta_1\) and easier once noted that the backward LTS imposes \(\kay(\theta_1) \neq \kay(\theta_2)\) for both transitions to be possible. 
\end{proof}

\begin{lemma}[Well-foundedness]
	\label{lem:wf}
	For all \(X\) there exists \(n \in \mathbb{N}\), \(X_0, \hdots, X_n\) \st  \(X \zigzagarrow X_n \zigzagarrow \cdots \zigzagarrow X_1 \zigzagarrow \orig{X}\) and \(\std(\orig{X})\).
\end{lemma}

This lemma forbids infinite reverse computation, and is obvious in \ccsk as any backward transition strictly decreases the number of occurrences of keys. %

\begin{proof}[\protect{of \autoref{thm:cc}}]
	We can re-use the results of the axiomatic approach~\cite{Lanese2020} since our forward LTS is the symmetric of our backward LTS, and as our concurrency relation (that the authors call the independence relation, following a common usage~\cite[Definition 3.7]{Sassone1996}) is indeed an irreflexive symmetric relation: symmetry is immediate by definition, irreflexivity follows from the fact that \(\dep\) is reflexive.
	Then, by \autoref{thm:sp} and \autoref{lem:bwconc}, the parabolic lemma holds~\cite[Proposition 3.4]{Lanese2020}, and since the parabolic lemma and well-foundedness hold (\autoref{lem:wf}), causal consistency holds as well~\cite[Proposition 3.5]{Lanese2020}. 
\end{proof}

\begin{example}
	Re-using the full trace presented in \autoref{ex1}, we can re-organize the transitions using the diamonds so that every undone transition is undone immediately, and we obtain up to causal equivalence the trace 	
	\[a . \out{b} \mid b +c \redl{\protect{\mid_{\L}a[m]}} a[m] . \out{b} \mid b + c \redl{\protect{\langle \mid_{\L}\out{b}[n], \mid_{\R}b[n]\rangle}} a[m] . \out{b}[n] \mid b[n] + c\]
\end{example}

\subsubsection{Optimality and Other Criteria}
An optimality criteria was defined in terms of adequation of the concurrency definitions for the LTS and for the reduction semantics~\cite[Theorem 5.6]{Cristescu2013}.
While this criteria requires a reduction semantics and a notion of reduction context to be formally proven, we believe it is easy to convince oneself that the gist of this property--the fact that non-\(\tau\)-transitions are concurrent iff there exists a \enquote{closing} context in which the resulting \(\tau\)-transitions are still concurrent--holds in our system: as concurrency on \(\tau\)-transitions is defined in terms of concurrency of its elements (\eg, \(\langle \theta^1_{\R}, \theta^1_{\L} \rangle \conc \langle \theta^2_{\R}, \theta^2_{\L} \rangle\) iff \(\theta^1_d \conc \theta^2_d\) for at least one \(d \in \{\L, \R\}\)), the optimality criteria is obtained \enquote{for free}.

Other properties~\cite[Section 6]{Cristescu2015b} have been investigated but would require a richer system to be truly meaningful: typically, \enquote{denotationality} requires for structural congruence to be preserved by the causal semantics, and hold vacuously in our structural-congruence-free system.
However, properties such as the parabolic lemma~\cite[Lemma 10]{Danos2004}--\enquote{any trace is equivalent to a backward trace followed by a forward trace}--or \enquote{RPI}~\cite[Definition 3.3]{Lanese2020}--\enquote{reversing preserves independence}, \ie concurrency is insensible to directions--follow immediately.

\section{Adding Infinite Behaviours to Reversible Systems}
\label{sec:infinite}
\subsection{Initial Considerations: Which to Pick?}
\label{ssec:discussion}

To endow \ccsk with infinite behaviors, the obvious choices are recursion, replication or iteration.
Carelessly adding the rules from \ccsd, \ccsr or \ccsi to \ccsk would result in the rules presented in \autoref{ccsk-infin} that, if \(X\) is not standard, would duplicate keys, voiding most of the propositions we have established so far.

\begin{figure}
	\begin{tcolorbox}[adjusted title=Infinite Group]
		\hfill
		\begin{prooftree}
			\hypo{X \fwlts{k}{\alpha}X'}
			\hypo{D \defeqc X}
			\infer[]2[const.]{D \fwlts{k}{\alpha} X'}
		\end{prooftree}
		\hfill
		\begin{prooftree}
			\hypo{X \fwlts{k}{\alpha} X'}
			\infer[]1[iter.]{X^* \fwlts{k}{\alpha} X' ; X^*}
		\end{prooftree}
		\hfill~
		\\[.9em]
		\hfill
		\begin{prooftree}
			\hypo{X \fwlts{k}{\alpha} X'}
			\infer[]1[repl.\(_1\)]{!X \fwlts{k}{\alpha} !X \mid X'}
		\end{prooftree}
		\hfill
		\begin{prooftree}
			\hypo{X \redl{\protect{\lambda[k]}} X'}
			\hypo{X \redl{\protect{\out{\lambda}[k]}} X''}
			\infer[]2[repl.\(_2\)]{!X \redl{\protect{\tau[k]}} !X \mid (X' \mid X'')}
		\end{prooftree}
		\hfill~
	\end{tcolorbox}
	\caption{Adding infinite behaviors to \ccsk}
	\label{ccsk-infin}
\end{figure}

For the replication and iteration operator, however, it is immediate to observe that \(!X\) or \(X^*\) are reachable if and only if \(\std(X)\): indeed, in iter.\ or the repl.\ rules, \(X\) is unchanged, and hence can never reduce \enquote{under} the \(!\) connector.
Similarly, if constants are required to be standard, then no memory can be duplicated using recursion: if they are not, then chances are that the LTS would disallow using the const.\ rule twice with the same constant, as it would allow to duplicate its keys.

It seems that all three forms of infinite behavior have \enquote{built-in mechanisms} to disallow duplication of keys in \ccsk, thus preventing the study of the replication of processes with a past.
On the other hand, \rccs seems to be able to accomodate duplication of the memory, at least on the surface of it~\cite[Section 4.1]{Aubert2021c}.
Indeed, one could have a process \(P\) endowed with a memory \(m\)--written \(m \vartriangleright P\)--use a rule similar to repl.\(_1\) to obtain \eg \(!(m \vartriangleright P) \efs !(m \vartriangleright P) \mid m' \vartriangleright P'\): while this term would certainly be problematic and require mechanisms to not duplicate the memory events in \(m\), or at least to disallow both threads from backtracking independently, \rccs's syntax certainly lends itself more naturally to express such behavior, that we reserve for future study.

Declensions of the const.\ rule, where \(X\) has to be standard, have been studied or mentioned for \rccs~\cite{Danos2004,Danos2005,Krivine2006} and \ccsk~\cite{Graversen2021}, but as far as we know, recursion and iterations have never been studied for distributed reversible systems.
As recursion, iteration and replication constitute a strict expressiveness hierarchy \wrt \eg weak bisimulation~\cite{Busi2009} in forward-only \ccs, we believe studying all three of them to understand if the same distinction holds in the presence of reversibility is a worthwhile task.
We propose to start this task by studying replication, as it is more commonly used when extending \ccs with name-passing abilities, in a way more convenient~\cite[p. 12]{Sangiorgi2001}, and as it has already been discussed for \rccs in an unpublished draft~\cite[Section 4.1]{Aubert2021c}.

\subsection{Adding Replication to \ccsk\ldots}
\label{ssec:repl-ccsk}
Adding replication to \ccsk requires the following adjustments:
\begin{enumerate}
	\item The replication operator \(!\) is added to \autoref{def:operators}, 
	\item In enhanced keyed labels (\autoref{def:enhanced-keyed-labels}), \(\upsilon\) can now range over \(\{\mid_{\L}, \mid_{\R}, !\}^*\), and \(\ell\) and \(\kay\) are extended straightforwardly,
	\item The dependency relation (\autoref{def:deprel}) now include

	\begin{minipage}{.2\textwidth}%
		\begin{align*}
		!\theta &\dep !\theta \\
		!\theta & \dep \mid_{\L} \theta' 
		\end{align*}%
	\end{minipage}%
	\hfill
	\begin{minipage}{.75\textwidth}
		\begin{numcases}{!\theta \dep \mid_{\R} \theta'\text{ if }}
			\theta = \langle \theta_{\L}, \theta_{\R}\rangle \text{ and } \theta' = \mid_{\L} \theta'' & or \label{a} \\
			\theta = \langle \theta_{\L}, \theta_{\R}\rangle \text{ and } \theta' = \mid_{\R} \theta'' & or \label{b} \\
			\theta \dep \theta' & otherwise,
		\end{numcases}
	\end{minipage}
	
	which immediately implies \(!\theta \dep \langle \mid_{\L} \theta_{\L}, \mid_{\R} \theta_{\R} \rangle\), as \(\theta \dep \langle \theta_{\L}, \theta_{\R} \rangle\) if \(\exists d . \theta \dep \theta_d\), but since \(!\theta \dep \mid_{\L} \theta_{\L}\) this is always the case.
	\label{dep-extd}
	\item The forward and backward proved LTS (\autoref{fig:provedltsrulesccskfw})  now have the rules
	\begin{center}
		\begin{prooftree}
			\hypo{X \redl{\theta} X'}
			\infer[right label={\(\kay(\theta) = k \)}]1[repl.\(_1\)]{!X \redl{!\theta} !X \mid X'}
		\end{prooftree}
		\hfill
		\begin{prooftree}
			\hypo{X \redl{\protect{\theta_{\L}\lambda[k]}} X'}
			\hypo{X \redl{\protect{\theta_{\R} \out{\lambda}[k]}} X''}
			\infer[]2[repl.\(_2\)]{!X \redl{\protect{!\langle \mid_{\L} \theta_{\L} \lambda \protect{[k]}, \mid_{\R} \theta_{\R} \out{\lambda} [k] \rangle}} !X \mid (X' \mid X'')}
		\end{prooftree}
	\end{center}
	and their reverses.
	Note that \(!X\) is reachable only if \(\std(X)\).
\end{enumerate}

\begin{remark}[On structural congruence, and the lack thereof]
	\label{rem:congruence}
	The equations (\ref{a}) and (\ref{b}) of the dependency relation comes from the absence of structural congruence: indeed, assuming \eg \(X \redl{\protect{\theta_{\L}}} X_1 \redl{\protect{\theta_{\L}'}} X_1'\) with \(\theta_{\L} \conc \theta_{\L}'\) and \(X \redl{\protect{\theta_{\R}}} X_2\), we cannot let \(!\langle \mid_{\L} \theta_{\L}, \mid_{\R} \theta_{\R}\rangle\) and \(\mid_{\R}\mid_{\L}\theta'_{\L}\) be concurrent in 
	\begin{align*}
	!X & \redl{\protect{!\langle \mid_{\L} \theta_{\L}, \mid_{\R} \theta_{\R}\rangle}} !X \mid (X_1 \mid X_2) \redl{\protect{\mid_{\R}\mid_{\L}\theta'_{\L}}} !X \mid (X_1' \mid X_2)\text{,}\\
	\shortintertext{as we cannot re-organize this trace.
	Indeed, while an extended version of the forward diamond lemma could allow us to obtain an \(X_1''\) and something like}
	!X & \redl{\protect{!\theta_{\L}'}} !X \mid X_1'' \redl{\protect{\langle \mid_{\R} ! \theta_{\R}, \mid_{\L} \theta_{\L} \rangle}} (!X \mid X_2) \mid X_1'\text{,}\\
	\shortintertext{we can \emph{not} make the target processes equal without a structural congruence making the parallel composition associative and commutative.}
	\end{align*}
	
	The \enquote{usual} structural congruence is missing from all the proved transition systems~\cite{Carabetta1998,Degano1992,Degano2001,Demangeon2018}, or missing the associativity and commutativity of the parallel composition~\cite[p.~242]{Degano1999}. %
	While adding such a congruence would benefits the expressiveness, making it interact nicely with the derived proof system \emph{and} the reversible features~\cites[Section 4]{Lanese2021}{Aubert2020d}~is a challenge we prefer to postpone. %
\end{remark}

Notions such as transition, trace, reachability and concurrency remain defined the same way, and \autoref{lem:part_of_trans} still holds, as it is not impacted by our changes.

\begin{example}
	Consider the transition \(X_1 = a.\out{a} + \out{a} \redl{\protect{a[m]}} a[m].\out{a} + \out{a} = X_2\), while we can obtain the derivation \(\pi_0\) below, it should be noted that \(!X_2\) is \emph{not} reachable and thus can be discarded.
	Indeed, as both repl.\ rules leave the term \(X\) unchanged, there is no transition that allows to reduce \enquote{under} the \(!\) connective.
	\begin{align*}
		\pi_0 & = 
		\begin{prooftree}
			\hypo{X_2 \redl{\protect{\out{a}[n]}} a[m].\out{a}[n] + \out{a}}
			\infer[]1[repl.\(_1\)]{! X_2 \redl{\protect{!\out{a}[n]}} ! X_2 \mid a[m].\out{a}[n] + \out{a}}
		\end{prooftree}  	
		\shortintertext{However, \(!X_1\) have two different \(\tau\)-transitions, corresponding respectively to \(!X_1\) synchronizing with its \enquote{future self} (\(\pi_1\)) or with itself (\(\pi_2\)).}
		\pi_1 & = 
		\begin{prooftree}
			\hypo{X_1 \redl{\protect{a[n]}} a[n].\out{a} + \out{a} }
			\infer[]1[repl.\(_1\)]{!X_1 \redl{\protect{!a[n]}} !X_1 \mid a[n].\out{a} + \out{a}}
			\hypo{X_2 \redl{\protect{\out{a}[n]}} a[m].\out{a}[n] + \out{a}}
			\infer[]2[syn.]{! X_1 \mid X_2 \redl{\protect{\langle \mid_{\L} !a[n], \mid_{\R}\out{a}[n]\rangle}} !X_1 \mid (a[n].\out{a} + \out{a} \mid a[m].\out{a}[n] + \out{a})}
		\end{prooftree}	
		\\[1.1em]
		\pi_2 & = 
		\begin{prooftree}
			\hypo{X_1 \redl{\protect{a[m]}} X_2 }
			\hypo{X_1 \redl{\protect{\out{a}[m]}} a.\out{a} + \out{a}[m]}
			\infer[]2[repl.\(_2\)]{!X_1 \redl{\protect{! \langle \mid_{\L} a[m], \mid_{\R}\out{a}[m]\rangle}} !X_1 \mid (X_2 \mid  a.\out{a} + \out{a}[m])}
		\end{prooftree}
	\end{align*}
\end{example}

\subsection{\ldots While Preserving Diamonds and Squares \ldots}
\label{ssec:square-inf}
We now revisit \autoref{sec:diamond}'s results and prove they carry over this extended calculus, up to a \emph{collapsing function} \(\sigma\) that identifies \(!\) and \(\mid_{\R}\) prefixes on labels: while those labels behave differently \wrt the dependency relation, they should be identified when re-organizing traces in \eg the extended forward diamond (\autoref{thm:forwardext}).

\begin{definition}[Collapsing function]
	We define, for \(d \in \{\L, \R\}\):
	\begin{align*}
		\sigma(\alpha[k]) & = \alpha[k]                  &   &   & \sigma(\mid_d \upsilon)                         & = \mid_d \sigma (\upsilon)                                  \\
		\sigma(!\upsilon) & = \mid_{\R} \sigma(\upsilon) &   &   & \sigma(\langle \theta_{\L}, \theta_{\R}\rangle) & = \langle \sigma(\theta_{\L}), \sigma(\theta_{\R}) \rangle 
	\end{align*}
\end{definition}

\begin{restatable}[Forward diamond (extd)]{theorem}{thmfwdiamondext}\label{thm:forwardext}
	For all \(X \redl{\theta_1} X_1 \redl{\theta_2} Y\) with \(\theta_1 \conc \theta_2\), there exists \(X_2\) \st  \(X \redl{\hat{\theta}_2} X_2 \redl{\hat{\theta}_1} Y\), and \(\sigma(\hat{\theta}_i) = \sigma(\theta_i')\) for \(i \in \{1, 2\}\).
\end{restatable}

\begin{proof}
	We can re-use the proof for \autoref{thm:forward} immediately, since \(\sigma (\theta) = \theta\) in all the cases we had, and since \(\theta_2\) cannot be of the form \(!\theta\) in the cases where we have to decompose it.
	But we now need to consider the case where repl.\(_1\) or repl.\(_2\) is the last rule of \(\theta_1\)'s derivation
	Those cases offer little resistance, but echoes back to \autoref{rem:congruence} to impose some dependencies between transitions that would require structural congruence to be able to re-order the resulting processes.
\end{proof}

\begin{restatable}[Sideways diamond (extd)]{theorem}{thmsidediamondext}\label{thm:sideext}
	For all \(X \redl{\theta_1} X_1 \revredl{\theta_2} Y\) with \(\theta_1 \conc \theta_2\), there exists \(X_2\) \st  \(X \revredl{\theta'_2} X_2 \redl{\theta'_1} Y\), and \(\sigma(\theta_i) = \sigma(\theta_i')\) for \(i \in \{1, 2\}\).
\end{restatable}

\begin{proof}
	As for the (non-extended) sideways diamond (\autoref{thm:side}), we can reuse the proof of the extended forward diamond almost as it is.
	An important case is if \(X\) is of the form \(!X'\), then \(!X' \redl{!\theta_1'} !X' \mid X_1'\) and \(\std(X')\), so the only transition \(X_1'\) can backtrack on is \(\theta_1'\).
	But \(\theta_2 = !\theta_1'\) or \(\mid_{\R}\theta_1'\) contradicts \(\theta_1 \conc \theta_2\), hence there are no backward transition following \(!\theta_1'\) not concurrent with it. %
\end{proof}

Finally, the loop lemma (\autoref{lemma:loop}) and hence the square property (\autoref{thm:sp})--under \(\sigma\)-image--are straightforward in our extended calculus, thanks to the same argument as in the non-extended case.

\subsection{\ldots as Well as Causal Consistency}
\label{ssec:causal-infinite}

Causal consistency (\autoref{thm:cc}) follows from the square property--established previously--, well-foundedness--which still holds, as \(!\) introduces infinite \emph{forward} behaviors but let the number of keys decreases strictly when backtracking--, but also that backward transitions are concurrent, a property we are missing in our extended calculus.
Indeed, consider \eg the reachable process \(X = !a.P \mid a[m].P\), we have \(X \revredl{\protect{!a[m]}} !a.P\) \emph{and} \(X \revredl{\protect{\mid_{\R} a[m]}} !a.P \mid a.P\), and since \(!a[m] \dep \mid_{\R} a[m]\), %
those two different traces are \emph{not} concurrent--this dependency was actually useful in proving \autoref{thm:sideext}.
There are two ways of restoring this needed property:
\begin{enumerate}
	\item One could argue that, since \(\sigma (!a[m]) = \mid_{\R}a[m]\), those two transitions are \emph{the same} up to the structural congruence \(!P \mid P \congru !P\)~\cite[Definition 12]{Busi2009},
	\item One could \enquote{mark} the key \(m\) so that it cannot backtrack independently, thus forbidding the \(\mid_{\R}a[m]\)-transition.
\end{enumerate}

We implement here the second solution, but would like the reader to appreciate that \emph{in all the development we carried so far, processes that can perform two conflicting backward transitions were not an issue}.
We believe this leeway could enable other formalisms or solutions to restore causal consistency.
Implementing this \enquote{marking} of keys resulting from a replication requires some adjustments:
\begin{enumerate}
	\item A key can now be either \emph{unmarked} \(\{m,n,\dots\}\) or \emph{marked} \(\{m^!,n^!,\dots\}\) 
	\item The forward and backward proved LTS (\autoref{fig:provedltsrulesccskfw})  now have the rules
	\begin{center}
		\begin{prooftree}
			\hypo{X \redl{\theta} X'}
			\infer[right label={\(\kay(\theta) = k \)}]1[repl.\(_1\)]{!X \redl{!\theta} !X \mid \markr(X')}
		\end{prooftree}
		\hfill
		\begin{prooftree}
			\hypo{X \redl{\protect{\theta_{\L}\lambda[k]}} X'}
			\hypo{X \redl{\protect{\theta_{\R} \out{\lambda}[k]}} X''}
			\infer[]2[repl.\(_2\)]{!X \redl{\protect{!\langle \mid_{\L} \theta_{\L} \lambda \protect{[k]}, \mid_{\R} \theta_{\R} \out{\lambda} [k] \rangle}} !X \mid \markr(X' \mid X'')}
		\end{prooftree}
	\end{center}
	and their inverses, for \(\markr(\alpha[m].P) = \alpha[m^!].P\), \(\markr(0) = 0\) and \(\markr(X)\) being simply \(\markr\) applied to the sub-term(s) of \(X\) recursively.
	\item Transitions \(X \revredl{\theta} X'\) are \emph{admissible} only if \(\kay(\theta)\) is un-marked in \(X'\), and we reason only on \emph{admissible} transitions in our lemmas and theorems.
\end{enumerate}

This forbids \enquote{just enough} to restore concurrency between all coinitial backward transitions, without altering anything else.
As a consequence, we can use again the axiomatic approach~\cite[Proposition 3.5]{Lanese2020} to obtain causal consistency (up to \(\sigma\)-image) for our extended system.
A possible drawback is that this solution comes at the price of compositionality: it is possible that \eg \(a[m^!].X\) does not have an admissible backward reduction, but \(!(a.X) \mid a[m^!].X\) does.
There seems to be no way of avoiding this \enquote{global} dimension of backtracking on replication.

\section{Conclusion and Perspectives}

Our two mains contributions are intertwined: infinite behaviors are a good test for concurrency definitions, which in turns should preserve the independence of \(!\)-transitions \wrt other transitions.
In this sense, finding definitions for concurrency and the replication operator that are suited for reversible systems was challenging, but we believe this work offers a solid, motivated answer to both questions.
By adapting a forward-only definition of concurrency, extending it to backward transitions \enquote{for free} and proving that expected properties such as the diamonds and squares hold, we provided an ingenious solution to a problem that have been left up in the air for too long.
Natural follow-ups include proving that this definition coincide with semantics-inspired~\cite{Graversen2021} and other syntactical~\cite[Definition 20]{Aubert2021d} representations of concurrency for other reversible systems.

Infinite behaviors in the presence of reversibility is not well-understood nor studied. %
Even if reversible primitive functions allowed a first exploration of the interplay of both notions~\cite{Paolini2018}, to our knowledge no systematic study of how replication can or cannot preserve \eg causal consistency had been led until this work.
Some attempts to extend algebras of communicating processes~\cite{Baeten2005}, including recursion, seems to have been unsuccessful~\cite{Wang2016}.
An interesting follow-up would be to define recursion and iteration as well in \ccsk, and to attempt to reconstruct the separation results from the forward-only paradigm~\cite{Palamidessi2005}: whether finer, \enquote{reversible}, equivalences can preserve this distinction despite the greater flexibility provided by backward transitions is an open problem.
Another interesting point is the study of infinite behaviors that duplicate past events, including their keys or memories: %
whether this formalism could preserve causal consistency, or what benefits there would be in tinkering this property, is also an open question.

Last but not least, both investigations would require to define relevant properties.
In the forward world, termination or convergence were used to compare infinite behavior~\cite{Palamidessi2005}, and additional criteria were introduced to study causal semantics~\cite{Cristescu2015b}.
Those properties may or may not be suited for reversible systems, but it is difficult to decide as they sometimes even lack a definition.
This could help in solving the more general question of deciding \emph{what} it is that we want to observe and assess when evaluating reversible, concurrent systems~\cite{Aubert2021h}.

\clearpage

\bibliographystyle{splncs04}
\bibliography{bib/bib}
\appendix

\section{Omitted Proofs}
\label{sec:proofs}

\lempartoftrans*

\begin{proof}
		The trace \(\pi_d(T)\) exists by virtue of the rule \(\mid_{d}\), syn.\ or their reverses.
		What remains to prove is that \(\pi_{d}(\theta_1) \conc_{\pi_{d}(T)} \pi_{d}(\theta_2)\) holds.
		
	The proof is by case on \(\theta_1\) and \(\theta_2\), but always follows the same pattern.
	As we know that both \(\pi_{d}(\theta_1)\) and \(\pi_{d}(\theta_2)\) need to be defined, there are 7 cases:

{\small
	\[
	\begin{tabu}{r || c | c |c | c |c | c |c | c |}
	\theta_1 & \mid_{\L}\theta_1' &  \mid_{\L}\theta_1' &  \mid_{\R}\theta_1' & \mid_{\R}\theta_1' & \langle \mid_{\L} \theta_1', \mid_{\R} \theta_1''\rangle & \langle \mid_{\L} \theta_1', \mid_{\R} \theta_1''\rangle & \langle \mid_{\L} \theta_1', \mid_{\R} \theta_1''\rangle \\
	\hline
	\theta_2 & \mid_{\L}\theta_2'& \langle \mid_{\L} \theta_2', \mid_{\R} \theta_2''\rangle & \mid_{\R}\theta_2'& \langle \mid_{\L} \theta_2', \mid_{\R} \theta_2''\rangle & \mid_{\L}\theta_2' & \mid_{\R}\theta_2'& \langle \mid_{\L} \theta_2', \mid_{\R} \theta_2''\rangle 
	\end{tabu}
	\]
}
By symmetry, we can bring this number down to three:
\[
\begin{tabu}{r || c | c | c | }
\text{(case letter)} & \textbf{a)} & \textbf{b)} & \textbf{c)} \\
\hline 
\theta_1 & \mid_{\L}\theta_1' & \langle \mid_{\L} \theta_1', \mid_{\R} \theta_1''\rangle & \langle \mid_{\L} \theta_1', \mid_{\R} \theta_1''\rangle\} \\
\hline
\theta_2 & \mid_{\L}\theta_2'& \mid_{\L}\theta_2'& \langle \mid_{\L} \theta_2', \mid_{\R} \theta_2''\rangle\}
\end{tabu}
\]
In each case, assume \(\pi_{\L}(\theta_1) = \theta_1' \conc_{\pi_{\L}(T)} \theta_2' = \pi_{\L}(\theta_2)\) does not hold.
Then it must be the case that either \(\theta_1' \dep_{\pi_{\L}(T)} \theta_{2}'\) or \(\theta_2' \dep_{\pi_{\L}(T)} \theta_{1}'\), and the two can be treated the same way, thanks to symmetry, and we detail this case for each of our three cases:
\begin{description}
	\item[a)] 
			If \(\theta_1' \dep_{\pi_{\L}(T)} \theta_{2}'\), then \(\theta_1' \dep \theta_{2}'\), and it is immediate that \(\theta_1 = \mid_{\L}\theta'_1 \dep_{T} \mid_{\L}\theta_2' = \theta_{2}\), contradicting \(\theta_1 \conc_{T} \theta_2\).
	\item[b)]
			If \(\theta_1' \dep_{\pi_{\L}(T)} \theta_{2}'\), then \(\theta_1' \dep \theta_{2}'\),  \(\mid_{\L}\theta_1' \dep \mid_{\L} \theta_{2}'\) and  \(\langle \mid_{\L}\theta_1', \mid_{\R}\theta''_1\rangle \dep \mid_{\L} \theta_{2}'\), from which we can deduce \(\theta_1 \dep_{T} \theta_{2}\), contradicting \(\theta_1 \conc_{T} \theta_2\).
	\item[c)]
			If \(\theta_1' \dep_{\pi_{\L}(T)} \theta_{2}'\), then \(\theta_1' \dep \theta_{2}'\),  \(\mid_{\L}\theta_1' \dep \mid_{\L} \theta_{2}'\) and  \(\langle \mid_{\L}\theta_1', \mid_{\R}\theta''_1\rangle \dep \langle \mid_{\L} \theta_{2}', \mid_{\R} \theta_2'\rangle\), from which we can deduce \(\theta_1 \dep_{T} \theta_{2}\), contradicting \(\theta_1 \conc_{T} \theta_2\).
\end{description}
Hence, in all cases, assuming that \(\pi_{d}(\theta_1) \conc_{\pi_{d}(T)} \pi_{d}(\theta_2)\) does not hold leads to a contradiction.%
\end{proof}

\thmfwdiamond*

\begin{proof}
	The proof proceeds by induction on the length of the deduction for the derivation for \(X \redl{\theta_1} X_1\).
	\begin{description}
		\item[Length \(1\)]
		In this case, the derivation is a single application of act., and \(\theta_1\) is of the form \(\alpha[k]\). 
		But \(\alpha[k] \conc \theta_2\) cannot hold, as \(\alpha[k] \dep \theta_2\) always holds, and this case is vacuously true.
		\item[Length \(> 1\)]
		We proceed by case on the last rule.
		\begin{description}
			\item[pre.] There exists \(\alpha\), \(k\), \(X'\) and \(X_1'\) \st  \(X = \alpha[k].X' \redl{\theta_1} \alpha[k].X_1' = X_1\) and that \(\kay(\theta_1) \neq k\).
			As \(\alpha[k].X_1' \redl{\theta_2} Y\) we know that \(\kay(\theta_2) \neq k\)~\cite[Lemma 3.4]{Lanese2021}, and we can apply \autoref{lem:rm} twice to obtain
			\[\rem^{\alpha}_k (\alpha[k].X') = X' \redl{\theta_1} \rem^{\alpha}_k (\alpha[k].X'_1) = X'_1 \redl{\theta_2} \rem^{\alpha}_k (Y)\]
			As \(\theta_1 \conc \theta_2\) by hypothesis, we can use the induction hypothesis to obtain that there exists \(X_2\) \st  \(X' \redl{\theta_2} X_2 \redl{\theta_1} \rem^{\alpha}_k (Y)\).
			Since \(\kay(\theta_2) \neq k\), we can append pre.\ to the derivation of \(X' \redl{\theta_2} X_2\) to obtain \(\alpha[k] . X' = X \redl{\theta_2} \alpha[k]. X_2\).
			Using \autoref{lem:rm} one last time, we obtain that \(\rem^{\alpha}_k(\alpha[k].X_2) = X_2 \redl{\theta_1} \rem^{\alpha}_k(Y)\) implies \(\alpha[k].X_2 \redl{\theta_1} Y\), which concludes this case.
			\item[res.] This is immediate by induction hypothesis.
			\item[\(\mid_{\L}\)]
			There exists \(X_{\L}\), \(X_{\R}\), \(\theta\), \(X_{1_{\L}}\), and \(Y_{\L}\), \(Y_{\R}\) \st  \(X \redl{\theta_1} X_1 \redl{\theta_2} Y\) is
			\[X_{\L} \mid X_{\R} \redl{\mid_{\L} \theta} X_{1_{\L}} \mid X_{\R} \redl{\theta_2} Y_{\L} \mid Y_{\R}\text{.}\]
			Then, \(X_{\L} \redl{\theta} X_{1_{\L}}\) and the proof proceeds by case on \(\theta_2\):
			\begin{description}
				\item[\(\theta_2\) is \(\mid_{\R} \theta'\)] %
				Then \(X_{\R} \redl{\theta'} Y_{\R}\), \(X_{1_{\L}} = Y_{\L}\) and the occurences of the rules \(\mid_{\L}\) and \(\mid_{\R}\) can be \enquote{swapped} to obtain 
				\[X_{\L} \mid X_{\R} \redl{\mid_{\R} \theta'} X_{\L} \mid Y_{\R} \redl{\mid_{\L}\theta} Y_{\L} \mid Y_{\R}\text{.}\]
				\item[\(\theta_2\) is \(\mid_{\L} \theta'\)] Then, \(X_{\L} \redl{\theta} X_{1_{\L}} \redl{\theta'} Y_{\L}\) and \(X_{\R} = Y_{\R}\).
				As \(\mid_{\L}\theta = \theta_1 \conc \theta_2 = \mid_{\L}\theta'\), it is the case that \(\theta \conc \theta'\) in \(X_{\L} \redl{\theta} X_{1_{\L}} \redl{\theta'} Y_{\L}\) by \autoref{lem:part_of_trans}, and we can use induction to obtain \(X_2\) \st  \(X_{\L} \redl{\theta'} X_2 \redl{\theta} Y_{\L}\), from which it is immediate to obtain \(X_{\L} \mid X_{\R} \redl{\mid_{\L} \theta'} X_2 \mid X_{\R} \redl{\mid_{\L}\theta} Y_{\L} \mid X_{\R} = Y_{\L} \mid Y_{\R}\)
				\item[\(\theta_2\) is \(\langle \mid_{\L} \theta_{\L}, \mid_{\R} \theta_{\R}\rangle\)] Since \(\mid_{\L}\theta = \theta_1 \conc \theta_2 = \langle \mid_{\L} \theta_{\L}, \mid_{\R} \theta_{\R}\rangle\), we have that \(\theta \conc \theta_{\L}\) in \(X_{\L} \redl{\theta} X_{1_{\L}} \redl{\theta_{\L}} Y_{\L}\) by \autoref{lem:part_of_trans}.
				Hence, we can use induction to obtain \(X_{\L} \redl{\theta_{\L}} X_2 \redl{\theta} Y_{\L}\).
				Since we also have that \(X_{\R} \redl{\theta_{\R}} Y_{\R}\), we can compose both traces using \emph{first} syn., then \(\mid_{\L}\) to obtain
				\[X_{\L} \mid X_{\R} \redl{\langle \mid_{\L} \theta_{\L}, \mid_{\R} \theta_{\R}\rangle} X_2 \mid Y_{\R} \redl{\mid_{\L}\theta} Y_{\L} \mid Y_{\R}\text{.}\]
			\end{description}
			\item[\(\mid_{\R}\)] This is symmetric to \(\mid_{\L}\).
			\item[syn.]
			There exists \(X_{\L}\), \(X_{\R}\), \(\theta_{\L}\), \(\theta_{\R}\), \(X_{1_{\L}}\), \(X_{1_{\R}}\), \(Y_{\L}\) and \(Y_{\R}\) \st  \(X \redl{\theta_1} X_1 \redl{\theta_2} Y\) is
			\[X_{\L} \mid X_{\R} \redl{\langle \mid_{\L} \theta_{\L}, \mid_{\R} \theta_{\R} \rangle} X_{1_{\L}} \mid X_{1_{\R}} \redl{\theta_2} Y_{\L} \mid Y_{\R}\text{.}\]
			Then, \(X_{\L} \redl{\theta_{\L}} X_{1_{\L}}\), \(X_{\R} \redl{\theta_{\R}} X_{1_{\R}}\) and the proof proceeds by case on \(\theta_2\):
			\begin{description}
				\item[\(\theta_2\) is \(\mid_{\R} \theta_{\R}'\)] %
				Then \(X_{1_{\R}} \redl{\theta_{\R}'} Y_{\R}\), \(X_{1_{\L}} = Y_{\L}\) and \(\langle \mid_{\L} \theta_{\L}, \mid_{\R} \theta_{\R} \rangle \conc \mid_{\R} \theta_{\R}'\) implies \(\theta_{\R} \conc \theta_{\R}'\) in \(X_{\R} \redl{\theta_{\R}} X_{1_{\R}} \redl{\theta_{\R}'} Y_{\R}\) by \autoref{lem:part_of_trans}.
				We can then use the induction hypothesis to obtain \(X_{\R} \redl{\theta_{\R}'} X_{2_{\R}} \redl{\theta_{\R}} Y_{\R}\) from which it is immediate to obtain \(X_{\L} \mid X_{\R} \redl{\mid_{\R} \theta_{\R}'} X_{\L} \mid X_{2_{\R}}  \redl{\langle \mid_{\L} \theta_{\L}, \mid_{\R} \theta_{\R} \rangle} X_{1_{\L}}  \mid Y_{\R} = Y_{\L} \mid Y_{\R}\).
				
				\item[\(\theta_2\) is \(\mid_{\L} \theta_{\L}'\)] This is symmetric to the previous one.
				\item[\(\theta_2\) is \(\langle \mid_{\L} \theta_{\L}', \mid_{\R} \theta_{\R}' \rangle\)] This case is essentially a combination of the two previous cases.
				 Since \(\langle \mid_{\L} \theta_{\L}', \mid_{\R} \theta_{\R}' \rangle = \theta_1 \conc \theta_2 = \langle \mid_{\L} \theta_{\L}, \mid_{\R} \theta_{\R}\rangle\), \autoref{lem:part_of_trans} gives us two traces
				 \begin{align*}
				 X_{\L} \redl{\theta_{\L}} X_{1_{\L}} \redl{\theta_{\L}'} Y_{\L}\qquad \text{and} \qquad 
 				 X_{\R} \redl{\theta_{\R}} X_{1_{\R}} \redl{\theta_{\R}'} Y_{\R}
				 \end{align*}
				 where \(\theta_{\L} \conc \theta_{\L}'\) and \(\theta_{\R} \conc \theta_{\R}'\), respectively.
				 By induction hypothesis, we obtain two traces 
				  \begin{align*}
				 X_{\L} \redl{\theta_{\L}'} X_{2_{\L}} \redl{\theta_{\L}} Y_{\L}\qquad \text{and} \qquad 
				 X_{\R} \redl{\theta_{\R}'} X_{2_{\R}} \redl{\theta_{\R}} Y_{\R}
				 \end{align*}
				 that we can then re-combine using syn.\ twice to obtain, as desired,
				 \[X_{\L} \mid X_{\R} \redl{\langle \mid_{\L} \theta_{\L}', \mid_{\R} \theta_{\R}' \rangle} X_{2_{\L}} \mid X_{2_{\R}} \redl{\langle \mid_{\L} \theta_{\L}, \mid_{\R} \theta_{\R} \rangle} Y_{\L} \mid Y_{\R}\text{.}\]
			\end{description}
			
			\item[\(+_{\L}\)] 
			There exists \(X_{\L}\), \(X_{\R}\), \(X_{\L}^1\), and \(Y_{\L}\) \st  \(X \redl{\theta_1} X_1 \redl{\theta_2} Y\) is
			\[X_{\L} + X_{\R} \redl{\theta_1} X_{\L}^1 + X_{\R} \redl{\theta_2} Y_{\L} + X_{\R}\text{.}\]
			Note that we know all transitions happen on \enquote{\(X_{\L}\)'s side} and \(X_{\R}\) remains unchanged as otherwise we could not sum two non-standard terms.
			Then, \(X_{\L} \redl{\theta_1} X_{\L}^1 \redl{\theta_2} Y_{\L}\), and as \(\theta_1 \conc \theta_2\) in this transition as well, we can use the induction hypothesis to obtain \(X_2\) \st \(X_{\L}  \redl{\theta_2} X_2 \redl{\theta_1} Y_{\L}\).
			From this, it is easy to obtain \(X_{\L} + X_{\R} \redl{\theta_2} X_2 + X_{\R} \redl{\theta_1} Y_{\L} + X_{\R}\) and this concludes this case.
			\item[\(+_{\R}\)] This is symmetric to \(+_{\L}\).
		\end{description}
	\end{description}
\end{proof}

\bwconc*

\begin{proof}
	The first important fact to note is that \(\kay(\theta_1) \neq \kay(\theta_2)\): by a simple inspection of the backward rules in \autoref{fig:provedltsrulesccskfw}, it is easy to observe that if a reachable process \(X\) can perform two different backward transitions, then they must have different keys.
	Then, the proof is by induction on the length of \(\theta_1\): while the lengths of the prefix and of the terms are not correlated, the result exposed in \autoref{app:sub-term}--and briefly reminded in the proof below-- allows to infer the structure of \(X\) from that of \(\theta_1\) and \(\theta_2\). 
\begin{description}
	\item[\(\theta_1\) is of length \(1\)]
	Then, \(\theta_1\) is a prefix and \(X\) is either 
	\begin{itemize}
		\item of the form \(\alpha[k].X'\) with \(\std(X')\), but in that case \(X\) cannot have two different backward transitions, so this case is vacuously true,
		\item of the form \(X' + Y\) with \(\std(Y)\) (or the symmetric), of the form \(X' \bs a\), or of the form \(\alpha[k].X'\) with \(X'\) not  standard: in all those cases, one can simply \enquote{extract} from the derivations the sub-trees that actually produce the \(\theta_1\) or \(\theta_2\) labels--as opposed to the rules that simply propagate it--, reason about it, and then restore the rest of the tree.
		As the concerned rules (\(\Rev{+_{\R}}\), \(\Rev{+_{\L}}\), \Rev{res.}\ and \Rev{pre.}) do not change the label, the dependencies or the derivability of the sub-tree, we can safely remove them, apply the previous reasoning to the \Rev{act.} rules that actually introduced the \(\theta_1\) and \(\theta_2\) labels, and restore the tree.
	\end{itemize}

	\item[\(\theta_1\) is of length greater than \(1\)]
		We proceed by case on the structure of \(\theta_1\), and use freely the previous remark--again, made formal in \autoref{app:sub-term}--that makes it possible to \enquote{skip over} \(+\), \(\bs a\) and \(\alpha[k]\) operators to focus on the sub-tree that \emph{actually} produces the labels: as an example, it is immediate that reasoning on a term \(X\) of the form \(Y_{\L} + Y_{\R}\) amounts to reasoning on \(Y_{\L}\) or \(Y_{\R}\), whichever is not standard, and that the other part of the process play no significant role in the proof.
	\begin{description}
		\item[\(\theta_1\) is of the form \(\mid_{\L}\theta_1'\)]
		
			Then we proceed by case on \(\theta_2\).
			\begin{description}
			\item[\(\theta_2\) is \(\mid_{\L}\theta_2'\)] Then, it holds by induction hypothesis after decomposing and re-composing both traces: we have
			\[X_{\L} \mid X_{\R} \revredl{\mid_{\L}\theta_1'} X_{1_{\L}} \mid X_{\R} \qquad \text{and}\qquad  X_{\L} \mid X_{\R} \revredl{\mid_{\L}\theta_2'} X_{2_{\L}} \mid X_{\R}\]
			and since \(\kay(\theta_1) \neq \kay(\theta_2)\), \(\theta_1' \neq \theta_2'\), so we can use the induction hypothesis to obtain that \(X_{\L} \revredl{\theta_1'} X_{1_{\L}}\) and \(X_{\L} \revredl{\theta_2'} X_{2_{\L}}\) are concurrent, from which we easily deduce that \(t_1: X\revredl{\theta_1}X_1\) and \(t_2:X\revredl{\theta_2}X_2\) are concurrent.
			\item[\(\theta_2\) is \(\mid_{\R}\theta_2'\)] Then, it is immediate, as \(\mid_{\L} \theta_1' \dep \mid_{\R}\theta_2'\) never hold.
			\item [\(\theta_2\) is \(\langle \mid_{\L}\theta_2', \mid_{\R} \theta_2''\rangle\)]
			Then we must prove that 
			\[X_{\L} \mid X_{\R} \revredl{\mid_{\L}\theta_1'} X_{1_{\L}} \mid X_{\R} \quad \text{and}\quad  X_{\L} \mid X_{\R} \revredl{\langle \mid_{\L}\theta_2', \mid_{\R} \theta_2''\rangle} X_{2_{\L}} \mid X_{2_{\R}}\]
			are concurrent.
			As we know that \(\kay(\theta_1) \neq \kay(\theta_2)\), \(\theta_1' \neq \theta_2'\), and we can use the induction hypothesis to get that 
			\[X_{\L} \revredl{\theta_1'} X_{1_{\L}} \qquad \text{and}\qquad  X_{\L} \revredl{\theta_2'} X_{2_{\L}}\text{,}\]
			are concurrent
			It is obvious that \(\mid_{\L}\theta_1' \dep \mid_{\R} \theta_2''\) does not hold, hence \(\mid_{\L}\theta_1' \conc \langle \mid_{\L} \theta_2', \mid_{\R} \theta_2''\rangle\) in \(X_{1_{\L}} \mid X_{\R} \redl{\mid_{\L}\theta} X_{\L} \mid X_{\R} \revredl{\langle \mid_{\L} \theta_2', \mid_{\R} \theta_2''\rangle} X_{2_{\L}} \mid X_{2_{\R}}\), and this concludes this case by \autoref{def:co-init-conc}.
		\end{description}	
		\item[\(\theta_1\) is of the form \(\mid_{\R}\theta_1'\)] This case is similar to the previous one.
		\item[\(\theta_1\) is of the form \(\langle \mid_{\L}\theta_1', \mid_{\R} \theta_1''\rangle\)] Then we proceed by case on \(\theta_2\) in a similar fashion, and this case offers no resistance.
\end{description}	
\end{description}	
\end{proof}

\thmfwdiamondext*

\begin{proof}
	We can re-use the proof for \autoref{thm:forward} immediately, since \(\sigma (\theta) = \theta\) in all the cases we had, and since \(\theta_2\) cannot be of the form \(!\theta\) in the cases where we have to decompose it.
	But we now need to consider the case where repl.\(_1\) or repl.\(_2\) is the last rule of \(\theta_1\)'s derivation.

	For repl.\(_1\), in a nutshell, if \(!X \redl{!\theta'_1} !X \mid X_1 \redl{\theta_2} Y_{\L} \mid Y_{\R}\) with \(!\theta'_1 \conc \theta_2\), then \(\theta_2\) cannot be of the form \(\mid_{\L}\theta_2'\) nor \(\langle \theta_{\L}, \theta_{\R}\rangle\), as  \(!\theta_1' \conc \theta_2\) would not hold in those cases, by \autoref{dep-extd} on p.~\pageref{dep-extd}.
	Since \(\theta_2\) cannot be of the form \(!\theta_2'\) either (as \(!X \mid X_1\) cannot make such a transition), it must be of the form \(\mid_{\R} \theta_2'\), with \(\theta_2' \conc \theta_1'\) in \(X \redl{\theta'_1} X_1 \redl{\theta_2'} Y_{\R}\)--otherwise, \(!\theta_1' = \theta_1 \conc \theta_2 = \mid_{\R}\theta_2'\) would not hold--, and hence \(Y_{\L} = !X\).
	By induction hypothesis, we obtain a transition \(X \redl{\hat{\theta}_2'} X_2 \redl{\hat{\theta}'_1} Y_{\R}\) with \(\sigma(\hat{\theta}_2') = \sigma(\theta_2')\) and \(\sigma(\hat{\theta}_1') = \sigma(\theta_1')\), from which it is easy to obtain \(!X \redl{!\hat{\theta}_2'} !X \mid X_2 \redl{\mid_{\R}\hat{\theta}_1'} !X \mid Y_{\R}\) as desired, with
	\begin{align*}
	\sigma(!\hat{\theta}_2') &= \mid_{\R}\sigma(\hat{\theta}_2')  = \mid_{\R} \sigma(\theta_2')  = \sigma(\mid_{\R} \theta_2')  = \sigma(\theta_2)
	\shortintertext{and}
	\sigma(\mid_{\R}\hat{\theta}'_1) & = \mid_{\R} \sigma(\hat{\theta}'_1) = \mid_{\R} \sigma(\theta_1') = \sigma(! \theta_1') = \sigma(\theta_1)
	\end{align*}
		
	For repl.\(_2\), if \(!X \redl{!\langle \mid_{\L} \theta_{\L}, \mid_{\R} \theta_{\R} \rangle } !X \mid (X_1 \mid X_2) \redl{\theta_2} Y_{\L} \mid Y_{\R}\), then observe that \(\theta_2\) cannot be of the form \(\mid_{\L}\theta_2'\) nor \(\langle \theta_{\L}', \theta_{\R}'\rangle \) without contradicting \(\theta_1 \conc \theta_2\), and that it cannot be of the form \(!\theta_2'\) neither as the top connector in \(!X \mid (X_1 \mid X_2)\) is not \(!\).
	Hence, it must be of the form \(\mid_{\R} \theta_{\R}'\), with \(\theta_{\R}'\) not of the form \(\mid_{\L}\theta_{\R}''\) nor \(\mid_{\R}\theta_{\R}''\) (following \autoref{rem:congruence}), nor \(!\theta_{R}''\) (as the top connector of \(X_1 \mid X_2\) is not \(!\)).
	Hence, \(\theta_2\) is of the form \(\mid_{\R} \langle \mid_{\L} \theta_{\L}', \mid_{\R} \theta_{\R}'\rangle\), and it follows by induction hypothesis: as \(\theta_1 \conc \theta_2\) in 
	\begin{align*}
	!X \redl{!\langle \mid_{\L} \theta_{\L}, \mid_{\R} \theta_{\R} \rangle } !X \mid (X_1 \mid X_2) \redl{\mid_{\R}\langle \mid_{\L} \theta_{\L}', \mid_{\R} \theta_{\R}' \rangle } !X \mid (X_1' \mid X_2')
	\shortintertext{it must be the case that \(\theta_{\L} \conc \theta_{\L}'\) and \(\theta_{\R} \conc \theta_{\R}'\) in}
	X \redl{\theta_L} X_1 \redl{\theta_{\L}'} X_1' \qquad \text{and}\qquad X \redl{\theta_R} X_2 \redl{\theta_{\R}'} X_2'
	\shortintertext{using the induction hypothesis twice gives us \(X_1''\), \(X_2''\), \(\hat{\theta}_{\L}\),  \(\hat{\theta}_{\L}'\), \(\hat{\theta}_{\R}\) and  \(\hat{\theta}_{\R}'\) \st }
	X \redl{\hat{\theta}_L'} X_1'' \redl{\hat{\theta}_{\L}} X_1' \qquad \text{and}\qquad  X \redl{\hat{\theta}'_R} X_2'' \redl{\hat{\theta}_{\R}} X_2'
	\shortintertext{from which we can obtain}
	!X \redl{!\langle \mid_{\L} \hat{\theta}'_{\L}, \mid_{\R} \hat{\theta}'_{\R} \rangle } !X \mid (X_1'' \mid X_2'') \redl{\mid_{\R}\langle \mid_{\L} \hat{\theta}_{\L}, \mid_{\R} \hat{\theta}_{\R} \rangle } !X \mid (X_1' \mid X_2')
	\end{align*}
	and as \(\sigma(!\theta) = \mid_{\R}\sigma(\theta)\), the condition on \(\sigma\) follows easily.
\end{proof}

\section{Technical Lemma on Sub-terms}
\label{app:sub-term}

This brief section exposes why we can, \wlg, assume that a term capable of performing a transition labeled by a prefix (\resp prefixed by \(\mid_{\L}\), \(\mid_{\R}\), \(\langle \theta_{\L}, \theta_{\R}\rangle\), \resp prefixed by \(!\)) can always be assumed to have for primary connector the same prefix (\resp a parallel composition, \resp a replication).
The lemma is not very insightful, but easy and critical in conducting the proof of \autoref{lem:bwconc}.

	\begin{definition}[Sub-term]
	Given \(X\), \(Y\), \(X\) is \emph{a subterm of \(Y\)}, \(X \subseteq Y\) if 
	\begin{align*}
		X = Y && X = Y + Y' && X = Y' + Y && X = Y \bs \alpha && X = a[k].Y
	\end{align*}
\end{definition}

\begin{lemma}
	\label{lem:decompose}
	If \(X \frevredl{\theta_1} Y\), then \(\exists X', Y'\) \st \(X' \subseteq X\), \(Y' \subseteq Y \), \(X' \frevredl{\theta_1} Y'\), and 
	\begin{itemize}
		\item if \(\theta_1\) is of the form \(\alpha[k]\), then the primary connector of \(X'\) is the prefix \(\alpha\) if the transition is forward, \(\alpha[k]\) if it is backward,
		\item if \(\theta_1\) is of the form \(\mid_{\L}\theta_1'\), \(\mid_{\R}\theta_1'\), or \(\langle \mid_{\L} \theta_{\L}, \mid_{\R} \theta_{\R}\rangle \), then the primary connector of \(X'\) is the parallel composition,
		\item if \(\theta_1\) is of the form \(! \theta_1'\), then the primary connector of \(X'\) is the replication.
	\end{itemize}
	Furthermore, if \(Y \frevredl{\theta_2} Z\) with \(\theta_1 \conc \theta_2\) in \(X \frevredl{\theta_1} Y \frevredl{\theta_2} Z\), then \(\exists Z'\) \st \(Y' \frevredl{\theta_2} Z'\) with \(Z' \subseteq Z\) and \(\theta_1 \conc \theta_2\) in  \(X' \frevredl{\theta_1} Y' \frevredl{\theta_2} Z'\).
\end{lemma}

\begin{proof}
	The first part is by case on \(\theta_1\), and on the primary connector of \(X\):
	\begin{description}
		\item[\protect{\(\theta_1\) is \(\alpha[k]\)}]
		We reason by induction on the size of \(X\), and sometimes on the direction of the transition:
		
		\begin{description}
			\item[If \(X\) is \(\beta.X''\)] Then it must be the case that \(\beta = \alpha\), otherwise \(X \frevredl{\protect{\alpha[k]}} Y\) would not hold, and we take \(X' = X\).
			\item[\protect{If \(X\) is \(\beta[m].X''\)}] If \(\beta[m].X'' \redl{\protect{\alpha[k]}} \beta[m].Y''\), then we apply the induction hypothesis to \(X''\redl{\protect{\alpha[k]}} Y''\), otherwise there are two cases: if \(X \revredl{\protect{\alpha[k]}} Y\) and \(\beta[m] \neq \alpha[k]\), then we use the induction hypothesis, otherwise we let \(X' = X\) and \(Y' = Y\).
			\item[If \(X\) is \(X_1 + X_2\)] Then if \(X_1 + X_2 \revredl{\protect{\alpha[k]}} X_1'' + X_2''\) with \(\std(X_1)\) (\resp \(\std(X_2)\)), we use the induction hypothesis on \(X_2 \revredl{\protect{\alpha[k]}} X_2''\) (\resp \(X_1 \revredl{\protect{\alpha[k]}} X_1''\)). If the transition is forward we use the induction hypothesis on \(X_1 \redl{\protect{\alpha[k]}} X_1'' \) (\resp \(X_2 \redl{\protect{\alpha[k]}} X_2''\)) if the last rule of the derivation for the transition is \(+_{\L}\) (\resp \(+_{\R}\)).
			\item[If \(X\) is \(X''\bs\alpha\)] Then \(Y = Y'' \bs \alpha\) and we use the induction on \(X'' \frevredl{\protect{\alpha[k]}} Y''\).
			\item[If \(X\) is \(X_{\L} \mid X_{\R}\)] Then \(X \frevredl{\protect{\alpha[k]}} Y\) is not possible.
			\item[If \(X\) is \(!X'\)] Then \(X \frevredl{\protect{\alpha[k]}} Y\) is not possible.
		\end{description}
		
		\item[\protect{\(\theta_1\) is \(\mid_{\L}\theta_1'\), \(\mid_{\R}\theta_1'\), or \(\langle \mid_{\L} \theta_{\L}, \mid_{\R} \theta_{\R}\rangle\)}]
		We reason by case on the primary connector of \(X\), and sometimes on the direction of the transition:
		
		\begin{description}
			\item[If \(X\) is \(\beta.X''\)] Then \(\beta.X'' \frevredl{\protect{\theta_1}} Y\) is not possible.
			\item[\protect{If \(X\) is \(\beta[m].X''\)}] If \(\beta[m].X'' \frevredl{\protect{\theta_1}} \beta[m].Y''\), then we apply the induction hypothesis to \(X'' \frevredl{\protect{\theta_1}} Y'' \).
			\item[If \(X\) is \(X_1 + X_2\) or \(X''\bs\alpha\)] This is the same reasoning as the sum and prefix cases whenever \(\theta_1 = \alpha[k]\), we apply the induction hypothesis to the subterm. %
			\item[If \(X\) is \(X_{\L} \mid X_{\R}\)] Then \(X' = X\).
			\item[If \(X\) is \(!X'\)] Then \(X \frevredl{\protect{\theta_1}} Y\) is not possible.
		\end{description}
		\item[\protect{\(\theta_1\) is \(!\theta_1'\)}]
		We reason by case on the primary connector of \(X\), and sometimes on the direction of the transition:
		
		\begin{description}
			\item[If \(X\) is \(\beta.X''\)] Then \(\beta.X'' \frevredl{\protect{!\alpha_1'}} Y\) is not possible.
			\item[\protect{If \(X\) is \(\beta[m].X''\)}] If \(\beta[m].X'' \frevredl{\protect{!\theta_1'}} \beta[m].Y''\), then we apply the induction hypothesis to \(X'' \frevredl{\protect{!\theta_1'}} Y'' \).
			\item[If \(X\) is \(X_1 + X_2\) or \(X''\bs\alpha\)] This is the same reasoning as the sum and prefix cases whenever \(\theta_1 = !\theta_1'\), we apply the induction hypothesis to the subterm. %
			\item[If \(X\) is \(X_{\L} \mid X_{\R}\)] Then \(X \frevredl{\protect{!\theta_1'}} Y\) is not possible.
			\item[If \(X\) is \(!X'\)]  Then \(X' = X\).
		\end{description}
		
	\end{description}
	For the second part, it amounts to observe that as \enquote{extracting} the sub-terms does not change the labels of the transitions, if  \(\theta_1 \conc \theta_2\) in \(X \frevredl{\theta_1} Y \frevredl{\theta_2} Z\), then applying the first part of this lemma to \(Y \frevredl{\theta_2} Z\) gives us the very same \(Y'\) as we obtain after applying the first part of this lemma to \(X \frevredl{\theta_1} Y\): we simply \enquote{go under} the \(+\), \(\bs \alpha\) and \(a[k]\) connectors whenever possible, and the exact same structure is stripped from the target of the transition. %
	The condition \(\theta_1 \conc \theta_2\) is required for the structures stripped away from \(X\) and \(Y\) to match: without this condition, we can get \eg \((a[m].X + b.Y)\bs c \revredl{\protect{a[m]}} (a.X + b.Y)\bs c \redl{\protect{b[m]}}  (a.X + b[m].Y) \bs c\) whose subterms \(a[m].X\) and \(b.Y\) are not equal, hence violating the second part of the lemma.
\end{proof}
\end{document}